\documentclass[12pt,reqno]{amsart}

\usepackage{amsmath}
\usepackage{amsthm}
\usepackage{amssymb}
\usepackage{graphicx}
\usepackage{fullpage}

\newcommand{\N}{\mathbb{N}}
\newcommand{\R}{\mathbb{R}}
\renewcommand{\S}{\mathbb{S}}

\newcommand{\Q}{\mathbb{Q}}

\renewcommand{\d}{\mathrm{d}}
\newcommand{\mc}{\mathcal}

\newcommand{\gT}{g_{\mathrm T}}

\newtheorem{lemma}{Lemma}[section]
\newtheorem{proposition}[lemma]{Proposition}
\newtheorem{theorem}[lemma]{Theorem}
\newtheorem{corollary}[lemma]{Corollary}
\theoremstyle{remark}
\newtheorem{remark}[lemma]{Remark}

\theoremstyle{definition}
\newtheorem{definition}[lemma]{Definition}

\numberwithin{equation}{section}
\numberwithin{table}{section}

\newcommand{\LR}[1]{{\langle #1 \rangle}}
\newcommand{\EQ}[1]{\begin{equation}\begin{split} #1 \end{split}\end{equation}}

\def\I{\infty}
\def\f{\frac}
\def\al{\alpha}
\def\eps{\varepsilon}
\def\nn{\nonumber}
\def\fy{\varphi}

\title{Linear stability of the Skyrmion}

\author{Matthew Creek}
\address{Department of Mathematics, University of Chicago, 5734 South University Avenue, Chicago, IL 60637, U.S.A.}
\email{mcreek@math.uchicago.edu}

\author{Roland Donninger}
\address{Rheinische Friedrich-Wilhelms-Universit\"at Bonn,
Mathematisches Institut, Endenicher Allee 60, D-53115 Bonn, Germany}
\email{donninge@math.uni-bonn.de}

\author{Wilhelm Schlag}
\address{Department of Mathematics, University of Chicago, 5734 South University Avenue, Chicago, IL 60637, U.S.A.}
\email{schlag@math.uchicago.edu}

\author{Stanley Snelson}
\address{Department of Mathematics, University of Chicago, 5734 South University Avenue, Chicago, IL 60637, U.S.A.}
\email{snelson@math.uchicago.edu}

\thanks{Roland Donninger is supported by the Alexander von Humboldt Foundation via
a Sofja Kovalevskaja Award endowed by the German Federal Ministry of Education
and Research. Partial support by the DFG, CRC 1060, is also gratefully acknowledged.
Furthermore, Roland Donninger would like to thank Pawe\l{} Biernat for many helpful discussions.
Mathew Creek and Stanley Snelson are partially supported by NSF grant DMS-1246999. Wilhelm Schlag is partially supported by NSF}

\begin{document}
\begin{abstract}
We give a rigorous proof for the linear stability of the Skyrmion.
In addition, we provide new proofs for the existence of the Skyrmion and the GGMT bound.
\end{abstract}

\maketitle

\section{Introduction}
\noindent In the 1960s and 1970s there was a lot of interest in classical relativistic nonlinear field theories as models for the interaction of elementary particles. The idea was to describe particles by solitons, i.e., static solutions of finite energy.
Due to the success of the standard model, where particles are described by \emph{linear} (but quantized) fields, this original motivation became somewhat moot.    
However, classical nonlinear field theories continue to be an active area of research, albeit for different reasons. 
They are interesting as models for Einstein's equation of general relativity,  in the context of nonperturbative quantum field theory or in the description of ferromagnetism. Furthermore, there is an ever-growing interest from the pure mathematical perspective.

A rich source for field theories with ``natural'' nonlinearities are geometric action principles. One of the most prominent examples of this kind is the SU(2) sigma model \cite{GelLev60} that arises from the wave maps action
\[ \mc S_{\mathrm{WM}}(u)=\int_{\R^{1,d}}\eta^{\mu\nu}(u^* g)_{\mu\nu}=\int_{\R^{1,d}}\eta^{\mu\nu}\partial_\mu u^A \partial_\nu u^B g_{AB}\circ u. \]
Here, the field $u$ is a map from $(1+d)$-dimensional Minkowski space $(\R^{1,d},\eta)$ to a Riemannian manifold $(M,g)$ with metric $g$.
Geometrically, the wave maps Lagrangian is the trace of the pull-back of the metric $g$ under the map $u$.
A typical choice is $M=\S^d$ with $g$ the standard round metric and in the following, we restrict ourselves to this case.
For $d=3$, one obtains the classical SU(2) sigma model. In general, the Euler-Lagrange equation associated to the action $\mc S_{\mathrm{WM}}$ is called the wave maps equation. 
Unfortunately, the SU(2) sigma model does not admit solitons and it develops singularities in finite time \cite{Sha88, BizChmTab00, Don11}.
One way to recover solitons is to lower the spatial dimension to $d=2$ but this is less interesting from a physical point of view and, even worse, the corresponding model still develops singularities in finite time \cite{BizChmTab01, KriSchTat08, RodSte10, RapRod12}.
Consequently, Skyrme \cite{Sky61} proposed to modify the wave maps Lagrangian by adding higher-order terms. This leads to the (generalized) Skyrme action \cite{ MakRybSan93}
\[ \mc S_{\mathrm{Sky}}(u)=\mc S_{\mathrm{WM}}(u)+\frac12 \int_{\R^{1,d}}\Big [
[\eta^{\mu\nu}(u^* g)_{\mu\nu}]^2-(u^* g)_{\mu\nu}(u^* g)^{\mu\nu} \Big ]. \]  
Skyrme's modification breaks the scaling invariance which makes the model more rigid.
Heuristically speaking, rigidity favors the existence of solitons and makes finite-time blowup less likely.
The original Skyrme model arises from the action $\mc S_{\mathrm{Sky}}$ in the case $d=3$ and $M=\S^3$.

By using standard spherical coordinates $(t,r,\theta,\varphi)$ on $\R^{1,3}$, one may consider so-called co-rotational maps $u: \R^{1,3}\to \S^3$ of the form $u(t,r,\theta,\varphi)=(\psi(t,r), \theta,\varphi)$. Under this symmetry reduction the Skyrme model reduces to the scalar quasilinear wave equation
\begin{equation}
\label{eq:Skyevol} 
(w\psi_t)_t-(w\psi_r)_r+\sin(2\psi)+\sin(2\psi)\left (\frac{\sin^2 \psi}{r^2}+\psi_r^2-\psi_t^2\right )=0 
\end{equation}
for the function $\psi=\psi(t,r)$, where $w=r^2+2\sin^2 \psi$.
It is well-known that there exists a static solution $F_0\in C^\infty[0,\infty)$ to Eq.~\eqref{eq:Skyevol} with the property that $F_0(0)=0$ and $\lim_{r\to\infty}F_0(r)=\pi$. This was proved by variational methods \cite{KapLad83} and ODE techniques \cite{McLTro91}.
In fact, $F_0$ is the \emph{unique} static solution with these boundary values \cite{McLTro91} and called the \emph{Skyrmion}.
Unfortunately, the Skyrmion is not known in closed form and
as a consequence, even the most basic questions concerning its role in the dynamics remain unanswered to this day.

\subsection{Stability of the Skyrmion}
Numerical studies \cite{BizChmRos07} strongly suggest that the Skyrmion is a global attractor for the nonlinear flow. In particular, $F_0$ should be stable under nonlinear perturbations. A first step in approaching this problem from a rigorous point of view is to consider the \emph{linear} stability of $F_0$.
To this end, one inserts the ansatz $\psi(t,r)=F_0(r)+\phi(t,r)$ into Eq.~\eqref{eq:Skyevol} and linearizes in $\phi$.
This leads to the linear wave equation
\[ \varphi_{tt}-\varphi_{rr}+\frac{2}{r^2}\varphi+V(r)\varphi=0 \]
for the auxiliary variable $\varphi(t,r)=\sqrt{r^2+2\sin^2F_0(r)}\,\phi(t,r)$.
The potential $V$ is given by
\[ V=-4a^2\frac{1+3a^2+3a^4}{(1+2a^2)^2},\qquad a(r)=\frac{\sin F_0(r)}{r}. \]
Consequently, the linear stability of the Skyrmion is governed by the $\ell=1$ Schr\"odinger operator
\[ \mc A f(r):=-f''(r)+\frac{2}{r^2}f(r)+V(r)f(r) \]
on $L^2(0,\infty)$.
More precisely, the Skyrmion is linearly stable if and only if $\mc A$ has no negative eigenvalues.
Unfortunately, the analysis of $\mc A$ is difficult since the potential $V$ is negative and not known explicitly.
Consequently, the linear stability of $F_0$ hinges on the particular shape of $V$ and this renders the application of general soft arguments hopeless.
Our main result is the following.

\begin{theorem}
\label{thm:main}
The Schr\"odinger operator $\mc A$ does not have eigenvalues.
In particular, the Skyrmion $F_0$ is linearly stable.
\end{theorem}

\subsection{Related work}
Due to the complexity of the field equation, there are not many rigorous results on dynamical aspects of the Skyrme model.
In \cite{GebNakRaj12}, small data global well-posedness and scattering is proved and \cite{Li12} establishes large-data global well-posedness. 
There is also some recent activity on the related but simpler Adkins-Nappi model, see e.g.~\cite{GebRaj10a, GebRaj10b, Law15}.
From a numerical point of view, the linear stability of the Skyrmion is addressed in \cite{HeuDroStr91} and
\cite{BizChmRos07} studies the nonlinear stability.
As far as the method of proof is concerned, we note that our approach is in parts inspired by \cite{CosHuaSch12}.

\subsection{Outline of the proof}
According to the GGMT bound, see \cite{GlaGroMarThi75, GlaGroMar78, ReeSim78} or Appendix \ref{app:GGMT}, the number of negative eigenvalues of $\mc A$ is bounded by
\[ \nu(V):=3^{-7}\frac{3^3\Gamma(8)}{4^4\Gamma(4)^2}\int_0^\infty r^7 |V(r)|^4\d r. \]
Consequently, our aim is to show that $\nu(V)<1$.
In fact, by a perturbative argument this also excludes the eigenvalue $0$
and there cannot be threshold resonances at zero energy since the decay of the recessive solution of $\mc A f=0$ is $1/r$ at infinity. 
In Appendix \ref{app:GGMT} we elaborate on this and give a new proof of the GGMT bound.

In order to show $\nu(V)<1$, we proceed by an explicit construction of the Skyrmion $F_0$. In particular, this yields a new proof for the existence of the Skyrmion. Our approach is mildly computer-assisted in the sense that one has to perform a large number of elementary operations involving fractions. It is worth noting that all computations are done in $\Q$, i.e., they are free of rounding or truncation errors. We also emphasize that the proof does not require a computer algebra system. Consequently, the necessary computations can easily be carried out using any programming language that supports fraction arithmetic. A natural choice is {\tt Python} which is open source and freely available for all common operating systems.

In the following, we give a brief outline of the main steps in the proof.

\begin{itemize}
\item We consider Eq.~\eqref{eq:Skyevol} for static solutions $\psi(t,r)=F(r)$ and change variables according to
\[ F(r)=2\arctan\left (r(1+r)g\left (\frac{r-1}{r+1}\right )\right ). \]
The new independent variable $x=\frac{r-1}{r+1}$ allows us to compactify the problem by considering $x\in [-1,1]$.
Furthermore, the $\arctan$ removes the trigonometric functions in Eq.~\eqref{eq:Skyevol}.
Consequently, we obtain an equation of the form
\[
 \mc R(g)(x):=g''(x)+\Phi(x,g(x),g'(x))=0 
 \]
where $\Phi$ is a (fairly complicated) rational function of $3$ variables.

\item We numerically construct a very precise approximation to the Skyrmion. This is done by employing a Chebyshev pseudospectral method \cite{Boy01}. The expansion coefficients are rationalized to allow for error-free computations in the sequel. This leads to a polynomial $\gT(x)$ with rational coefficients and we rigorously prove that
$\|\mc R(\gT)\|_{L^\infty(-1,1)}\leq \frac{1}{500}$.
As a consequence, the construction of the Skyrmion reduces to finding a (small) correction $\delta(x)$ such that $\mc R(\gT+\delta)=0$.

\item Next, we obtain bounds on second derivatives of $\Phi$ by employing rational interval arithmetic. As a consequence, we obtain the representation
\[ \mc R(\gT+\delta)=\mc R(\gT)+\mc L \delta+\mc N(\delta) \]
with explicit bounds on the nonlinear remainder $\mc N$.
The linear operator $\mc L$ is also given explicitly in terms of $\gT$ and first derivatives of $\Phi$.

\item Again, by a Chebyshev pseudospectral method, we numerically construct an approximate fundamental system $\{u_-,u_+\}$ for the linear equation $\mc Lu=0$.
The functions $u_\pm$ satisfy $\tilde{\mc L}u_\pm=0$ for another linear operator $\tilde{\mc L}$ that is close to $\mc L$ in a suitable sense.
Using $u_\pm$ we construct an inverse $\tilde{\mc L}^{-1}$ to $\tilde{\mc L}$ which allows us to rewrite the equation $\mc R(\gT+\delta)=0$ as a fixed point problem
\[ \delta=-\tilde{\mc L}^{-1}\mc R(\gT)-\tilde{\mc L}^{-1}(\mc L-\tilde{\mc L})\delta-\tilde{\mc L}^{-1}\mc N(\delta)=:\mc K(\delta). \]
From the explicit form of $u_\pm$ we obtain rigorous and explicit bounds on the operator $\tilde{\mc L}^{-1}$.

\item Finally, we prove that $\mc K$ is a contraction on a small closed ball in $W^{1,\infty}(-1,1)$.
This yields the existence of a small correction $\delta(x)$ such that $\gT+\delta$ solves the transformed Skyrmion equation.
From the uniqueness of the Skyrmion we conclude that
\[ F_0(r)=2\arctan\left (r(1+r)(\gT+\delta)\left (\frac{r-1}{r+1}\right )\right ) \]
and the desired $\nu(V)<1$ follows by elementary estimates.

\end{itemize}

\subsection{Notation}
Throughout the paper we abbreviate $L^\infty:=L^\infty(-1,1)$ and also $W^{1,\infty}:=W^{1,\infty}(-1,1)$. For the norm in $W^{1,\infty}$ we use the convention
\[ \|f\|_{W^{1,\infty}}:=\sqrt{\|f'\|_{L^\infty}^2+\|f\|_{L^\infty}^2}. \]
The Wronskian $W(f,g)$ of two functions $f$ and $g$ is defined as $W(f,g):=fg'-f'g$.

\section{Preliminary transformations}

\noindent Static solutions $\psi(t,r)=F(r)$ of Eq.~\eqref{eq:Skyevol} satisfy the Skyrmion equation 
\begin{equation}
\label{eq:sky} \frac{\d}{\d r}\Big [\big (r^2+2\sin^2 F(r)\big )F'(r)\Big ]-\sin(2 F(r))\left [
F'(r)^2+\frac{\sin^2 F(r)}{r^2}+1 \right ]=0.
\end{equation}
The Skyrmion $F_0$ is the unique solution of Eq.~\eqref{eq:sky} satisfying $F_0(0)=0$ and
$\lim_{r\to\infty}F_0(r)=\pi$.
More precisely, we have $F_0(r)=\pi+O(r^{-2})$ as $r\to\infty$.
Furthermore, it is known that the Skyrmion is monotonically increasing \cite{McLTro91}.
In order to remove the trigonometric functions it is thus natural to define a new dependent variable 
$f: [0,\infty)\to \R$ by 
\[ F(r)=:2\arctan f(r). \]
Then we have
\begin{align*}
F'=\frac{2f'}{1+f^2},\qquad F''=\frac{2f''}{1+f^2}-\frac{4f'^2 f}{(1+f^2)^2}
\end{align*}
as well as
\[ \sin^2 F=\frac{4f^2}{(1+f^2)^2},\qquad \sin(2 F)=\frac{4f(1-f^2)}{(1+f^2)^2}. \]
Consequently, Eq.~\eqref{eq:sky} is equivalent to
\begin{equation}
\label{eq:skyf}
f''+\frac{\mc W(f)'}{\mc W(f)}f'-\frac{2f'^2f}{1+f^2}-\frac{2f(1-f^2)}{\mc W(f)(1+f^2)}\left [
\frac{4f'^2}{(1+f^2)^2}+\frac{4f^2}{r^2(1+f^2)^2}+1\right ]=0
\end{equation}
where 
\[ \mc W(f)(r):=r^2+\frac{8f(r)^2}{[1+f(r)^2]^2}. \]
Eq.~\eqref{eq:skyf} may be slightly simplified to give
\begin{equation}
\label{eq:skyf2}
f''+\frac{2rf'}{\mc W(f)}-\frac{2f'^2f}{1+f^2}+\frac{2f(1-f^2)}{\mc W(f)(1+f^2)}\left [
\frac{4f'^2}{(1+f^2)^2}-\frac{4f^2}{r^2(1+f^2)^2}-1\right ]=0
\end{equation}
Next, we set
\[ f(r)=:r(1+r)g\left (\frac{r-1}{r+1}\right ). \]
This yields
\begin{align*}
f\left (\frac{1+x}{1-x}\right )&=2\frac{1+x}{(1-x)^2}g(x) \\
f'\left (\frac{1+x}{1-x}\right )&=(1+x)g'(x)+\frac{3+x}{1-x}g(x) \\
f''\left (\frac{1+x}{1-x}\right )&=\tfrac12(1+x)(1-x)^2g''(x)+2(1-x)g'(x)+2g(x)
\end{align*}
for $x\in [-1,1)$.
We compactify the problem by allowing $x\in [-1,1]$.
In these new variables, Eq.~\eqref{eq:skyf} can be written as
\begin{equation}
\label{eq:skyg}
\mc R(g)(x):=g''(x)+\Phi\big (x,g(x),g'(x)\big )=0
\end{equation}
where $\Phi: (-1,1)\times \R^2\to \R$ is given by
\begin{align}
\label{def:Phi}
\Phi(x,y,z):=\frac{1}{\Psi(x,y)}\sum_{k=0}^2\Phi_k(x,y)z^k
\end{align}
with
\begin{align}
\label{def:Phik}
\Phi_0(x,y)&:=2^{-5}(1+x)^5(3+x)y^7
-2^{-6}(1+x)(1-x)^3(33-58x-16x^2+18x^3+7x^4)y^5 \nonumber \\
&\qquad +2^{-9}(1-x)^7(47-51x+33x^2+3x^3)y^3
+2^{-9}(1-x)^{11}y \nonumber \\
\Phi_1(x,y)&:=-2^{-4}(1+x)^7y^6-2^{-5}(1+x)^2(1-x)^4(14-21x+4x^2+7x^3)y^4 \nonumber \\
&\qquad +2^{-8}(1-x)^8(23-31x+13x^2+3x^3)y^2
+2^{-9}(1-x)^{12} \nonumber \\
\Phi_2(x,y)&:=-(1-x^2)\big [2^{-5}(1+x)^6y^5
+2^{-6}(1+x)^2(1-x)^4(7-10x+7x^2)y^3 \nonumber \\
&\qquad -2^{-9}(1-x)^8(3-10x+3x^2)y \big ]
\end{align}
and
\begin{align}
\label{def:Psi}
 \Psi(x,y):=&(1-x^2)\big [2^{-6}(1+x)^6y^6
+2^{-8}(1+x)^2(1-x)^4(11-10x+11x^2)y^4 \nonumber \\
&+2^{-10}(1-x)^8(11-10x+11x^2)y^2
+2^{-12}(1-x)^{12}\big ] .
\end{align}
Obviously, $\Psi(-1,y)=\Psi(1,y)=0$ for all $y$
and, since
\begin{align}
\label{eq:Psi1boundary}
 \sum_{k=0}^2 \Phi_k(-1,y)z^k&=4(1+8y^2)(y+2z) \nonumber \\
  \sum_{k=0}^2 \Phi_k(1,y)z^k&=4y^6(y-2z), 
 \end{align}
we obtain the 
 regularity conditions
 \begin{equation}
 \label{eq:RegRem}
  g'(-1)=-\tfrac12 g(-1),\qquad g'(1)=\tfrac12 g(1) 
  \end{equation}
for solutions of $\mc R(g)=0$ (at least if $g(1)\not=0$, which is the case we
are interested in).

\section{Numerical approximation of the Skyrmion}

\subsection{Description of the numerical method}
\label{sec:num}
We will require a fairly precise approximation to the Skyrmion.
Already from a numerical point of view this is not entirely trivial since a brute force
approach is doomed to fail.
That is why we
employ a more sophisticated Chebyshev pseudospectral method.
To this end, we use the basis functions $\phi_n: [-1,1]\to \R$, $n\in \N_0$, 
given by
\begin{equation}
\label{def:phi}
 \phi_n(x):=T_n(x)+a_n(1+x)+b_n(1-x),
 \end{equation}
where $T_n$ are the standard Chebyshev polynomials.
The constants $a_n$ and $b_n$ are chosen in such a way that the regularity conditions Eq.~\eqref{eq:RegRem}
are satisfied, i.e., we require
\begin{equation}
\label{eq:Regphi}
 \phi_n'(-1)+\tfrac12 \phi_n(-1)=\phi_n'(1)-\tfrac12\phi_n(1)=0 
 \end{equation}
for all $n\in \N_0$.
This yields $\phi_0=\phi_1=0$ and
\begin{align*}
 a_n&=-T_n'(-1)-\tfrac12 T_n(-1)=(-1)^{n}(n^2-\tfrac12) \\
 b_n&=T_n'(1)-\tfrac12 T_n(1)=n^2-\tfrac12
\end{align*}
for $n\geq 2$.
Then we numerically solve the ($N_0-1$)-dimensional nonlinear root finding problem
\[ \mc R\left (\sum_{n=2}^{N_0}\tilde c_n\phi_n\right )(x_k)=0, \quad x_k=\cos\left (\frac{k\pi}{N_0}\right ),
\quad k=1,2,\dots,N_0-1 \]
for $N_0=43$ with $\mc R$ given in Eq.~\eqref{eq:skyg}.
The points $(x_k)_{k=1}^{N_0-1}$ are the standard Gau\ss-Lobatto collocation points 
for the Chebyshev pseudospectral method \cite{Boy01} with endpoints removed (we only have $N_0-1$ unknown
coefficients due to $\phi_0=\phi_1=0$; in the standard Chebyshev method one has $N_0+1$ coefficients
to determine).
Finally, we rationalize the numerically obtained coefficients $(\tilde c_n)$.
The $42$ coefficients $(c_n)_{n=2}^{43}\subset \Q$ obtained in this way are listed in Table \ref{tab:c}.

\subsection{Methods for rigorous estimates}

In order to obtain good estimates for the complicated rational functions that will show up in the sequel,
the following elementary observation is useful.

\begin{lemma}
\label{lem:estf}
Let $f\in C^1([-1,1])$ and set
\[ \Omega_N:=\{-1+\tfrac{2k}{N}: k=0,1,2,\dots,N\}\subset [-1,1]\cap \Q,\qquad N\in \N . \]
Then we have the bounds
\begin{align*}
\max_{[-1,1]}f&\leq \max_{\Omega_N}f+\tfrac{2}{N}\|f'\|_{L^\infty} \\
\min_{[-1,1]}f&\geq \min_{\Omega_N}f-\tfrac{2}{N}\|f'\|_{L^\infty} \\
\|f\|_{L^\infty}&\leq \max_{\Omega_N}|f|+\tfrac{2}{N}\|f'\|_{L^\infty}
\end{align*}
for any $N\in \N$.
\end{lemma}

\begin{proof}
The statements are simple consequences of the mean value theorem.
\end{proof}

\begin{remark}
In a typical application one first obtains a rigorous but crude bound on $f'$ by elementary estimates. Then one uses a computer to evaluate $f$ sufficiently many times in order to obtain a good bound on $f$. 
\end{remark}

Another powerful method for estimating complicated functions is provided by interval arithmetic \cite{AleMay00, HicJuEmd01}. 
We use the following elementary rules for operations involving intervals.
\begin{definition}
Let $a,b,c,d\in \R$ with $a\leq b$ and $c\leq d$. \emph{Interval arithmetic} is defined by the following operations.
\begin{align*}
[a,b]+[c,d]&:=[a+c, b+d] \\
[a,b]-[c,d]&:=[a-d, b-c] \\
[a,b]\cdot [c,d]&:=[\min\{ac, ad, bc, bd\}, \max\{ac, ad, bc, bd\}] \\
\frac{[a,b]}{[c,d]}&:=[a,b]\cdot [\tfrac{1}{d}, \tfrac{1}{c}]\quad \mbox{provided }0\notin [c,d].
\end{align*}
If $a,b,c,d\in \Q$, we speak of \emph{rational interval arithmetic}. 
Furthermore, standard (rational) arithmetic is embedded by identifying $a\in \R$ with $[a,a]$. 
\end{definition}

\begin{lemma}
\label{lem:interval}
Let $x\in [a,b]$ and $y\in [c,d]$ and denote by $*$ any of the elementary operations $+, -, \cdot, /$. Then we have $x*y\in [a,b]*[c,d]$.
\end{lemma}

\begin{proof}
The proof is an elementary exercise.
\end{proof}

\begin{remark}
\label{rem:interval}
If $f$ is a complicated rational function of several variables (with rational coefficients), rational interval arithmetic is an effective way to obtain a rigorous and reasonable bound on $f(\Omega)$, provided $\Omega$ is a product of closed intervals with rational endpoints.
The necessary computations can easily be carried out on a computer as they only involve elementary operations in $\Q$.
The quality of the bound, however, depends on the particular algebraic form that is used to represent $f$.
Furthermore, in typical applications the bound can be improved considerably by splitting the domain $\Omega$ in smaller subdomains $\Omega_k$, i.e., $\Omega=\bigcup_k \Omega_k$, and by estimating each $f(\Omega_k)$ separately by interval arithmetic.
\end{remark}

\subsection{Rigorous bounds on the approximate Skyrmion}

\begin{definition}
We set
\[ \gT(x):=\sum_{k=2}^{43}c_n\phi_n(x) \]
where $(c_n)_{n=2}^{43}\subset \Q$ are given in Table \ref{tab:c}.
\end{definition}

\begin{proposition}
\label{prop:Rem}
The function $\gT$ satisfies
\begin{align*} \tfrac{1}{100}+\tfrac{11}{20}&\leq \gT(x)
\leq \tfrac{21}{20}-\tfrac{1}{100} \\
\tfrac{1}{100}-\tfrac{11}{20}&\leq \gT'(x)\leq \tfrac{1}{2}-\tfrac{1}{100}
\end{align*}
for all $x\in [-1,1]$.
Furthermore,
\[ \|\mc R(\gT)\|_{L^\infty}\leq \tfrac{1}{500}. \]
\end{proposition}

\begin{proof}
From the bound $\|T_n''\|_{L^\infty}\leq \tfrac13 n^2(n^2-1)$ we infer
\[ \|\gT''\|_{L^\infty}\leq \sum_{n=2}^{43}|c_n|\|T_n''\|_{L^\infty}\leq 
\tfrac13\sum_{n=2}^{43}n^2(n^2-1)|c_n|\leq 36  \]
and Lemma \ref{lem:estf} with $N=7200$ yields
\begin{align*} \max_{[-1,1]} \gT'&\leq \max_{\Omega_N} \gT'+\tfrac{2}{N}\|\gT''\|_{L^\infty}
\leq \tfrac{47}{100}+\tfrac{1}{100}\leq \tfrac{1}{2}-\tfrac{1}{100}  \\
\min_{[-1,1]}\gT'&\geq \min_{\Omega_N}\gT'-\tfrac{2}{N}\|\gT''\|_{L^\infty}
\geq -\tfrac{51}{100}-\tfrac{1}{100}\geq -\tfrac{11}{20}+\tfrac{1}{100}.
\end{align*}
In particular, we obtain $\|\gT'\|_{L^\infty}\leq 1$ and
with $N=200$ we find
\begin{align*} \max_{[-1,1]} \gT&\leq \max_{\Omega_N} \gT+\tfrac{2}{N}\|\gT'\|_{L^\infty}
\leq \tfrac{101}{100}+\tfrac{1}{100}\leq \tfrac{21}{20}-\tfrac{1}{100}  \\
\min_{[-1,1]}\gT&\geq \min_{\Omega_N}\gT-\tfrac{2}{N}\|\gT'\|_{L^\infty}
\geq \tfrac{58}{100}-\tfrac{1}{100}\geq \tfrac{11}{20}+\tfrac{1}{100}.
\end{align*}
This proves the first part of the Proposition.

Next, we consider
\[ \hat\Psi(x,y):=\frac{\Psi(x,y)}{1-x^2}. \]
Rational interval arithmetic yields 
\[ \hat\Psi\left ([-1,0],[\tfrac{11}{20},\tfrac{21}{20}]\right )\subset \left [10^{-3},13\right ],\qquad
\hat\Psi\left ([0,1],[\tfrac{11}{20},\tfrac{21}{20}]\right )\subset \left [10^{-4},2\right]
 \]
 and thus, $\hat \Psi(x,\gT(x))>0$ for all $x\in [-1,1]$. 
We set
\begin{align*} P(x)&:=\frac{(\frac{21}{10}+\frac13 x-x^2)^7 }{1-x^2}\sum_{k=0}^2
\Phi_k\left (x,\gT(x)\right )[\gT'(x)]^k \\
\qquad Q(x)&:=\frac{(\tfrac{21}{10}+\tfrac13 x-x^2)^7}{1-x^2}\Psi\left (x,\gT(x)\right )=(\tfrac{21}{10}+\tfrac13 x-x^2)^7\hat\Psi(x,\gT(x)), 
\end{align*}
which yields the representation
\[ \Phi\left (x,\gT(x),\gT'(x)\right )=\frac{P(x)}{Q(x)}. \]
The prefactor $(\frac{21}{10}+\frac13 x-x^2)^7$ is introduced 
\emph{ad hoc}. It is empirically found to improve some of the estimates that follow.
By Eq.~\eqref{def:Psi}, $Q$ is a polynomial with rational coefficients and by the regularity
conditions Eq.~\eqref{eq:Regphi} together with Eq.~\eqref{eq:Psi1boundary}, the same is true
for $P$.
Furthermore, $Q(x)>0$ for all $x\in [-1,1]$
and from the explicit expressions for $\Phi_k$ and $\Psi$, 
Eqs.~\eqref{def:Phik} and \eqref{def:Psi}, we read
off the estimates $\deg P\leq 319$ and $\deg Q\leq 278$.

For the following it is advantageous to straighten the denominator.
To this end we obtain a truncated Chebyshev expansion of $1/Q$,
\[ \frac{1}{Q(x)}\approx \sum_{n=0}^{14} r_n T_n(x)=:R(x), \]
where 
\[ (r_n)=(\tfrac{11}{37},-\tfrac{1}{23}, -\tfrac{5}{44}, -\tfrac{3}{13},
\tfrac{9}{44}, \tfrac{1}{12}, -\tfrac{1}{766}, -\tfrac{3}{25}, \tfrac{1}{101},
\tfrac{1}{23}, \tfrac{1}{35}, -\tfrac{1}{36}, -\tfrac{1}{66}, \tfrac{1}{307}, \tfrac{1}{125}). \]
The coefficients $(r_n)$ can be obtained numerically by a standard pseudospectral method as explained in Section \ref{sec:num}.
Thus, we may write
\begin{align*}
 \mc R(\gT)(x)&=\gT''(x)+\Phi\left (x,\gT(x),\gT'(x)\right )=\gT''(x)+\frac{P(x)}{Q(x)} \\
 &=\frac{R(x)Q(x)\gT''(x)+R(x)P(x)}{R(x)Q(x)}
 \end{align*}
and this modification is expected to improve the situation since the denominator $RQ$ is now
approximately constant. Note further that $RP$ and $RQ$ are polynomials with rational
coefficients and 
\[ \deg(RP)\leq 333,\quad \deg (RQ)\leq 292,\quad \deg(RQ\gT'')\leq 333. \]
For brevity we set
\[ \hat P:=RQ\gT''+RP,\qquad \hat Q:=RQ. \]

We now re-expand $\hat P$ and $\hat Q$ as
\[ \hat P(x)=\sum_{n=0}^{333}\hat p_nT_n(x),
\qquad \hat Q(x)=\sum_{n=0}^{292}\hat q_nT_n(x). \]
The expansion coefficients $(\hat p_n),(\hat q_n)\subset \Q$ are obtained
by solving the linear equations\footnote{The choice of the evaluation points $(x_k)$ is arbitrary but since $\hat P$ has
removable singularities at $-1$ and $1$, we prefer to avoid the endpoints.
Furthermore, the equation for $(\hat q_n)$ is overdetermined so that one can re-use the computationally expensive LU decomposition.}
\begin{align*} 
\sum_{n=0}^{333}\hat p_nT_n(x_k)=\hat P(x_k),\qquad 
\sum_{n=0}^{333}\hat q_nT_n(x_k)=\hat Q(x_k),\qquad 
x_k=-\tfrac12+\tfrac{k}{333}
\end{align*}
for $k=0,1,\dots,333$.
From the bounds $\|T_n\|_{L^\infty}\leq 1$ and $\|T_n'\|_{L^\infty}\leq n^2$ we infer
\begin{align*}
\|\hat P\|_{L^\infty}\leq \sum_{n=0}^{333}|\hat p_n|
\leq \tfrac{12}{10000}, \qquad 
\|\hat Q'\|_{L^\infty}\leq\sum_{n=0}^{292}n^2|\hat q_n|\leq 22.
\end{align*}
Consequently, Lemma \ref{lem:estf} with $N=500$ yields
\begin{align*} 
\min_{[-1,1]} \hat Q\geq \min_{\Omega_N}\hat Q-\tfrac{2}{N}\|\hat Q'\|_{L^\infty}\geq \tfrac{93}{100}-\tfrac{44}{500}\geq \tfrac{4}{5}
\end{align*}
and, since $\mc R(\gT)=\hat P/\hat Q$, we obtain the estimate
\[ \|\mc R(\gT)\|_{L^\infty}\leq 
\frac{\|\hat P\|_{L^\infty}}{\min_{[-1,1]}\hat Q}\leq \tfrac54\tfrac{12}{10000}=\tfrac{3}{2000}\leq \tfrac{4}{2000}=\tfrac{1}{500}. \]
\end{proof}

\section{Estimates for the nonlinearity}

\noindent By employing rational interval arithmetic, we prove bounds on second derivatives of the function $\Phi$.
This leads to explicit bounds for the nonlinear operator.

All of the polynomials of two variables $x, y$ that appear in the sequel are implicitly assumed to be given in the following \emph{canonical form}
\[ \sum_{k=0}^{k_0}(1+x)^{\alpha_k}(1-x)^{\beta_k}P_k(x)y^k \]
where $k_0,\alpha_k,\beta_k \in \N_0$ and $P_k$ are polynomials with rational coefficients and $P_k(\pm 1)\not= 0$.
This is important since the outcome of interval arithmetic depends on the representation of the function.

\subsection{Pointwise estimates}

\begin{lemma} 
\label{lem:pwN}
Let $\Omega=[-1,1]\times [\tfrac{11}{20},\tfrac{21}{20}]\times [-\tfrac{11}{20},\tfrac12]$.
Then we have the bounds
\begin{align*}
\|\partial_2^2 \Phi\|_{L^\infty(\Omega)}&\leq 70 \\
\|\partial_2\partial_3 \Phi\|_{L^\infty(\Omega)}&\leq 22 \\
\|\partial_3^2 \Phi\|_{L^\infty(\Omega)}&\leq 8. 
\end{align*}
\end{lemma}

\begin{proof}
We begin with the simplest estimate, that is, the bound on $\partial_3^2\Phi$.
We set
\[ \hat \Phi_k(x,y):=\frac{\Phi_k(x,y)}{1-x^2},\qquad \hat\Psi(x,y):=\frac{\Psi(x,y)}{1-x^2} \]
with $\Phi_k$ and $\Psi$ from Eqs.~\eqref{def:Phik} and \eqref{def:Psi}, respectively.
Observe that $\hat\Phi_2$ is a polynomial.
From Eq.~\eqref{def:Phi} we infer
\[ \partial_z^2 \Phi(x,y,z)=\frac{2\Phi_2(x,y)}{\Psi(x,y)}
=\frac{2\hat\Phi_2(x,y)}{\hat\Psi(x,y)} \]
and from the proof of Proposition \ref{prop:Rem} we recall
that $\hat\Psi([-1,1],[\frac{11}{20},\frac{21}{20}])\subset [10^{-4},13]$.
Consequently, $\partial_3^2 \Phi$ is a rational function without poles in $\Omega$.
Rational interval arithmetic then yields\footnote{Here and in the following, the domain $\Omega$ needs to be divided in sufficiently small subdomains $\Omega_k\subset \Omega$ such that $\Omega=\bigcup_k \Omega_k$, see Remark \ref{rem:interval}.}
$\partial_3^2\Phi(\Omega) \subset [-8,8]$
and this proves the stated bound for $\partial_3^2\Phi$.

Next, we consider $\partial_2\partial_3\Phi$.
We have
\begin{align*} \partial_y\partial_z\Phi(x,y,z)&=\partial_y \frac{\hat \Phi_1(x,y)+2\hat \Phi_2(x,y)z}{\hat \Psi(x,y)} \\
&=\frac{\hat\Psi(x,y)\partial_y\hat\Phi_1(x,y)
-\partial_y\hat\Psi(x,y)\hat\Phi_1(x,y)}{\hat\Psi(x,y)^2} \\
&\quad+2z\frac{\hat\Psi(x,y)\partial_y\hat\Phi_2(x,y)
-\partial_y\hat\Psi(x,y)\hat\Phi_2(x,y)}{\hat\Psi(x,y)^2} 
\end{align*}
and, since $\hat\Phi_2$ is a polynomial, the last term is a rational function without
poles in $\Omega$.
Note further that the numerator of the second to last term appears to be singular
at $x\in \{-1,1\}$, but in fact there is a cancellation so that
\begin{align*} \hat\Psi(x,y)&\partial_y\hat\Phi_1(x,y)-\partial_y\hat\Psi(x,y)\hat\Phi_1(x,y) \\
&=2^{-11}(1+x)^7(1-x)^3(17-43x+7x^2+3x^3)y^9 \\
&\quad -2^{-11}(1+x)^5(1-x)^7(17-15x+7x^2+7x^3)y^7 \\
&\quad -2^{-14}(1+x)(1-x)^{11}(285-637x+794x^2-386x^3+41x^4+95x^5)y^5 \\
&\quad -2^{-15}(1+x)(1-x)^{15}(25-31x+15x^2+7x^3)y^3 \\
&\quad +2^{-19}(1-x)^{19}(1-12x+3x^2)y.
\end{align*}
We conclude that $\partial_2\partial_3\Phi$ is a rational function
without poles in $\Omega$
and rational interval arithmetic yields
$\partial_2\partial_3\Phi(\Omega)\subset [-22, 22]$.

Finally, we turn to $\partial_2^2 \Phi$.
We have
\begin{align*} 
\partial_y \Phi(x,y,z)&=\sum_{k=0}^2 
\frac{\hat\Psi(x,y)\partial_y\hat\Phi_k(x,y)z^k
-\partial_y\hat\Psi(x,y)\hat\Phi_k(x,y)z^k}{\hat\Psi(x,y)^2} \\
&=\frac{1}{\hat\Psi(x,y)^2}\sum_{k=0}^2 \hat\Psi_k(x,y)z^k
\end{align*}
where $\hat\Psi_k:=\hat\Psi\partial_2\hat\Phi_k-\partial_2\hat\Psi\hat\Phi_k$.
From above we recall that $\hat\Psi_1$ and $\hat\Psi_2$ are polynomials.
We obtain
\begin{align*}
\partial_y^2 \Phi(x,y,z)=\sum_{k=0}^2 
\frac{\hat\Psi(x,y)^2\partial_y\hat\Psi_k(x,y)z^k
-2\hat\Psi(x,y)\partial_y\hat\Psi(x,y)\hat\Psi_k(x,y)z^k}{\hat\Psi(x,y)^4}.
\end{align*}
Again, the apparently singular term 
\[ \hat\Psi(x,y)^2\partial_y\hat\Psi_0(x,y)
-2\hat\Psi(x,y)\partial_y\hat\Psi(x,y)\hat\Psi_0(x,y) \]
is in fact a polynomial since it exhibits a special cancellation. 
Consequently, $\partial_2^2\Phi$ is a rational function
without poles in $\Omega$ and
rational interval arithmetic yields the desired bound.
\end{proof}

\subsection{The nonlinear operator}
In this section we employ Einstein's summation convention, i.e., we sum over repeated indices (the range follows from the context).

\begin{lemma}
\label{lem:Taylor}
Let $U\subset \R^d$ be open and convex and $f\in C^2(U)\cap W^{2,\infty}(U)$.
Set 
\[ M:=\tfrac12\left (\sum_{j=1}^d\sum_{k=1}^d \|\partial_j\partial_k f\|_{L^\infty(U)}^2 \right )^{1/2}. \]
Then we have 
\[ f(x_0+x)=f(x_0)+x^j\partial_j f(x_0)+N(x_0, x) \]
where $N$ satisfies the bound
\begin{align*}
 |N(x_0, x)-N(x_0, y)|&\leq M(|x|+|y|)|x-y|
\end{align*}
for all $x_0, x, y\in \R^d$ such that $x_0, x_0+x, x_0+y\in U$.
\end{lemma}

\begin{proof}
From the fundamental theorem of calculus we infer
\begin{align*}
 N(x_0,x)-N(x_0,y)&=f(x_0+x)-f(x_0+y)-(x^j-y^j)\partial_jf(x_0) \\
 &=\int_0^1 \partial_t f\big (x_0+y+t(x-y)\big )\d t-(x^j-y^j)\partial_jf(x_0) \\
 &=(x^j-y^j)\int_0^1 \big [\partial_j f\big (x_0+y+t(x-y)\big )-\partial_j f(x_0)\big ]\d t \\
 &=(x^j-y^j)\int_0^1\int_0^1 \partial_s \partial_j f\big (x_0+sy+st(x-y)\big )\d s \d t \\
 &=(x^j-y^j)\int_0^1[y^k+t(x^k-y^k)]\int_0^1 \partial_k\partial_j f\big (x_0+sy+st(x-y)\big )\d s \d t 
\end{align*}
and Cauchy-Schwarz yields
\begin{align*}
|N(x_0,x)-N(x_0,y)|&\leq |x^j-x^j|\|\partial_j\partial_k f\|_{L^\infty(U)}
\int_0^1 \big [t|x^k|+(1-t)|y^k|\big ]\d t \\
&=\tfrac12 |x^j-y^j|(|x^k|+|y^k|)\|\partial_j\partial_k f\|_{L^\infty(U)} \\
&\leq \tfrac12 |x-y|(|x^k|+|y^k|)\left (\sum_{j=1}^d \|\partial_j\partial_k f\|_{L^\infty(U)}^2\right )^{1/2} \\
&\leq M|x-y||x|+M|x-y||y|.
\end{align*}
\end{proof}

\begin{proposition}
\label{prop:CN}
We have
\[ \mc R(\gT+\delta)=\mc R(\gT)+\mc L\delta+\mc N(\delta) \]
where
\[ \mc Lu(x):=u''(x)+\partial_3\Phi\big(x,\gT(x),\gT'(x)\big )u'(x)+
\partial_2\Phi\big (x,\gT(x),\gT'(x)\big )u(x)
 \]
and $\mc N$ satisfies the bounds
\begin{align*}
\|\mc N(u)\|_{L^\infty}&\leq 39\,\|u\|_{W^{1,\infty}}^2 \\
\|\mc N(u)-\mc N(v)\|_{L^\infty}
&\leq 39\left (\|u\|_{W^{1,\infty}}+\|v\|_{W^{1,\infty}}\right )
\|u-v\|_{W^{1,\infty}}
\end{align*}
for all $u,v \in C^1[-1,1]$ with
$\|u\|_{W^{1,\infty}}, \|v\|_{W^{1,\infty}}\leq \frac{1}{100}$.
\end{proposition}

\begin{proof}
Let $\Omega=[-1,1]\times [\tfrac{11}{20},\tfrac{21}{20}]\times [-\tfrac{11}{20},\tfrac12]$. Lemma \ref{lem:Taylor} implies
\[ \Phi(x,y_0+y,z_0+z)=\Phi(x, y_0, z_0)
+\partial_2 \Phi(x, y_0, z_0)y+\partial_3 \Phi(x, y_0, z_0)z
+N(x, y_0, z_0, y, z) \]
where $N$ satisfies the bound
\begin{align*} |N(x, y_0, z_0, y, z)-N(x,y_0,z_0,\tilde y,\tilde z)|&\leq 
M\sqrt{(y-\tilde y)^2+(z-\tilde z)^2}
\left (\sqrt{y^2+z^2}+\sqrt{\tilde y^2+\tilde z^2}\right )
\end{align*}
with
\[ M=\tfrac12\sqrt{\|\partial_2^2\Phi\|_{L^\infty(\Omega)}^2
+2\|\partial_2\partial_3 \Phi\|_{L^\infty(\Omega)}^2
+\|\partial_3^2\Phi\|_{L^\infty(\Omega)}^2}. \]
From Lemma \ref{lem:pwN} we infer $M\leq 39$ and thus,
the claim follows 
from Proposition \ref{prop:Rem} by
setting
\[ \mc N(u)(x):=N\big (x, \gT(x), \gT'(x), u(x), u'(x)\big ). \]
\end{proof}

\section{Analysis of the linear operator}

\noindent In this section we construct a linear operator $\tilde{\mc L}$ with an explicit fundamental system such that $\mc L-\tilde{\mc L}$ is small in $L^\infty(-1,1)$. Then we invert $\tilde{\mc L}$ and prove an explicit bound on the inverse. 

\subsection{Asymptotics}
First, we study the asymptotic behavior
of $\partial_2\Phi$ and $\partial_3 \Phi$.

\begin{lemma}
\label{lem:asymPhi}
We have
\begin{align*}
\partial_2\Phi\big (x,\gT(x),\gT'(x)\big )&=\frac{2}{1+x}+O(x^0) \\
\partial_3\Phi\big (x,\gT(x),\gT'(x)\big )&=\frac{4}{1+x}+O(x^0)
\end{align*}
 for $x\in (-1,0]$,
as well as
\begin{align*}  
\partial_2\Phi\big (x,\gT(x),\gT'(x)\big)&=\frac{2}{1-x}+O(x^0) \\
\partial_3\Phi\big (x,\gT(x),\gT'(x)\big)&=-\frac{4}{1-x}+O(x^0)
\end{align*}
for $x\in [0,1)$.
\end{lemma}

\begin{proof}
As before, we set
\[ \hat\Psi(x,y):=\frac{\Psi(x,y)}{1-x^2} \]
with $\Psi$ from Eq.~\eqref{def:Psi}.
Then we have
\[ \Phi(x,y,z)=\frac{1}{(1-x^2)\hat\Psi(x,y)}\sum_{k=0}^2 \Phi_k(x,y)z^k \]
with $\Phi_k$ given in Eq.~\eqref{def:Phik}.
Recall that $\hat\Psi$ is a polynomial with no zeros in $[-1,1]\times [\frac{11}{20},\frac{21}{20}]$, see the proof of Proposition \ref{prop:Rem}.
From Eqs.~\eqref{def:Phik} and \eqref{def:Psi} we obtain
\begin{align*}
\Phi_0(-1,y)&=4y+32y^3 & \Phi_0(1,y)&=4y^7 \\
\Phi_1(-1,y)&=8+64y^2 & \Phi_1(1,y)&=-8y^6 \\
\Phi_2(-1,y)&=0 & \Phi_2(1,y)&=0 \\
\hat\Psi(-1,y)&=1+8y^2 & \hat\Psi(1,y)&=y^6.
\end{align*} 
Consequently,
\begin{align*}
 \lim_{x\to -1}\left [(1+x)\partial_z \Phi(x,y,z)\right ]
&=\frac{\Phi_1(-1,y)}{2\hat\Psi(-1,y)}=4 \\
\lim_{x\to 1}\left [(1-x)\partial_z \Phi(x,y,z)\right ]
&=\frac{\Phi_1(1,y)}{2\hat\Psi(1,y)}=-4 .
\end{align*}
The other assertions are proved similarly.
\end{proof}

In order to isolate the singular behavior it is natural to write 
\[ \mc Lu=\mc L_0u +pu'+qu \]
where
\begin{align*}
 \mc L_0u(x)&=u''(x)+\left (\frac{4}{1+x}-\frac{4}{1-x}\right )u'(x)+\left (\frac{2}{1+x}+\frac{2}{1-x}\right )u(x)  \\
 &=u''(x)-\frac{8x}{1-x^2}u'(x)+\frac{4}{1-x^2}u(x) \\
 p(x)&=\partial_3\Phi\big (x,\gT(x),\gT'(x)\big )-\frac{4}{1+x}+\frac{4}{1-x} \\
 q(x)&=\partial_2\Phi\big (x,\gT(x),\gT'(x)\big )-\frac{2}{1+x}-\frac{2}{1-x}. 
 \end{align*}
Lemma \ref{lem:asymPhi} implies that $p$ and $q$ are rational functions with no poles in $[-1,1]$.

\begin{lemma}
\label{lem:fsasym}
The equation $\mc Lu=0$ has fundamental systems $\{u_-,v_-\}$ and $\{u_+,v_+\}$
on $(-1,1)$
which satisfy
\begin{align*} 
u_{-}(x)&=1+O(1+x) \\
u_{-}'(x)&=-\tfrac12+O(1+x) \\
 v_{-}(x)&=O((1+x)^{-3}) 
 \end{align*}
for $x\in (-1,0]$, as well as
\begin{align*}
 u_+(x)&=1+O(1-x) \\
 u_+'(x)&=\tfrac12+O(1-x) \\
  v_+(x)&=O((1-x)^{-3}) 
  \end{align*}
for $x \in [0,1)$.
Furthermore, $u_-,v_-,u_+,v_+\in C^\infty(-1,1)$ and $u_- \in C^\infty([-1,1))$,
$u_+\in C^\infty((-1,1])$.
\end{lemma}

\begin{proof}
The coefficients of the equation $\mc Lu=0$ are rational functions and the only poles
in $[-1,1]$ are at $x=-1$ and $x=1$.
These poles are regular singular points of the equation with Frobenius indices $\{-3,0\}$.
Consequently,
the statements follow by Frobenius' method.
\end{proof}

\subsection{Numerical construction of an approximate fundamental system}

We obtain an approximate fundamental system $\{u_-,u_+\}$, where $u_\pm$ is smooth at $\pm 1$, by a Chebyshev pseudospectral method. 
As always, special care has to be taken near the singular endpoints $\pm 1$.
Solutions $u$ of $\mc Lu=0$ that are regular at $-1$ must satisfy $u'(-1)+\frac12 u(-1)=0$. Similarly, regularity at $1$ requires $u'(1)-\frac12 u(1)=0$, cf.~Eq.~\eqref{eq:RegRem}.
If one sets
\[ u_\pm(x)=\frac{w_\pm(x)}{(1\pm x)^3}, \] 
the regularity conditions $u_\pm'(\pm 1)=\pm \frac12 u_\pm(\pm 1)$ translate into $w_\pm'(\pm 1)=\pm 2w_\pm(\pm 1)$.
Consequently, we use the basis functions $\psi_{\pm,n}: [-1,1]\to \R$, $n\in \N$, given by
\begin{align}
\label{def:psi}
\psi_{\pm,n}(x)&:=T_n(x)\pm [T_n'(\pm 1)\mp 2T_n(\pm 1)](1\mp x)
\end{align} 
which have the necessary regularity conditions automatically built in, i.e.,
$\psi_{\pm,n}'(\pm 1)=\pm 2\psi_{\pm,n}(\pm 1)$ for all $n\in \N$.
Observe that $w_\pm$ is expected to be bounded on $[-1,1]$, see Lemma \ref{lem:fsasym}.
For brevity, we also set
\begin{equation}
\hat\psi_{\pm,n}(x):=\frac{\psi_{\pm,n}(x)}{(1\pm x)^3}.
\end{equation}
We enforce the normalization
\[ \sum_{n=1}^{N_\pm}c_{\pm,n}\hat\psi(\pm 1)=1, \]
which is used to fix the coefficients $c_{\pm, 1}$.
The remaining coefficients are obtained numerically by solving the root finding problem
\[ \mc L\left (\sum_{n=1}^{N_\pm}c_{\pm,n}\hat\psi_{\pm,n}\right )(x_k)=0,\quad 
x_k=\cos\left (\frac{k\pi}{N_\pm}\right ),\quad k= 1,2,\dots,N_\pm-1 \]
with $N_\pm=30$.
Finally, we rationalize the floating-point coefficients.
The resulting coefficients are listed in Tables \ref{tab:cm} and \ref{tab:cp}.

\subsection{Rigorous bounds on the approximate fundamental system}
The numerical approximation leads to the following definition.

\begin{definition}
\label{def:fs}
We set
\[ u_\pm(x):=\frac{w_\pm(x)}{(1\pm x)^3}:=\frac{1}{(1\pm x)^3}\sum_{n=1}^{30} c_{\pm,n}\psi_{\pm,n}(x) \]
where
the coefficients $(c_{\pm,n})_{n=2}^{30}\subset \Q$ are given in Tables \ref{tab:cm} and \ref{tab:cp}, respectively.
The coefficients $c_{\pm,1}$ are determined by the requirement $u_\pm(\pm 1)=1$.
\end{definition}

Next, we analyze the approximate fundamental system $\{u_-,u_+\}$.

\begin{proposition}
\label{prop:fs}
We have $W(u_-,u_+)(x)=(1-x^2)^{-4} W_0(x)$, where $W_0$ is a polynomial with no zeros in $[-1,1]$.
Furthermore, the functions $u_\pm$ satisfy
\[ \tilde{\mc L}u_{\pm}=0, \]
where $\mc{\tilde L} u:=\mc L_0u+\tilde pu'+\tilde q u$, and
\[ \|\tilde p-p\|_{L^\infty}\leq \tfrac{3}{100},\qquad \|\tilde q-q\|_{L^\infty}\leq \tfrac{1}{20}. \]
\end{proposition}

\begin{proof}
We temporarily set $p_\pm(x):=(1\pm x)^{-3}$.
Then we have
\[ W(u_-,u_+)=W(p_- w_-,p_+ w_+)=W(p_-,p_+)w_-w_++p_-p_+W(w_-,w_+) \]
and, since $W(p_-,p_+)(x)=-6(1-x^2)^{-4}$, we infer
$W(u_-,u_+)(x)=(1-x^2)^{-4} W_0(x)$ with
\[ W_0(x)=-6 w_-(x)w_+(x)+(1-x^2)W(w_-,w_+)(x). \]
Obviously, $W_0$ is a polynomial with $\deg W_0\leq 61$, see Definition \ref{def:fs}.
We re-expand $W_0$ in Chebyshev polynomials,
\[ W_0(x)=\sum_{n=0}^{61} w_{0,n}T_n(x), \]
by solving the (possibly overdetermined) system
\[ \sum_{n=0}^{61} w_{0,n}T_n(x_k)=W_0(x_k),\quad x_k=-\tfrac12+\tfrac{k}{61}, \quad k=0,1,2,\dots,61 \]
for the coefficients $(w_{0,n})_{n=0}^{61}\subset \Q$. 
From the re-expansion we obtain the estimate
\[ \|W_0'\|_{L^\infty}\leq \sum_{n=0}^{61} |w_{0,n}|\|T_n'\|_{L^\infty}
\leq \sum_{n=0}^{61} n^2 |w_{0,n}|\leq 400 \] 
and Lemma \ref{lem:estf} with $N=2000$ yields
\[ \max_{[-1,1]}W_0\leq \max_{\Omega_N}W_0+\tfrac{2}{N}\|W_0'\|_{L^\infty}
\leq -\tfrac{94}{100}+\tfrac{400}{1000}\leq -\tfrac12. \]
This shows that $W_0$ has no zeros in $[-1,1]$.

We set
\[ \tilde p:=\frac{u_+\mc L_0u_- - u_-\mc L_0 u_+}{W(u_-,u_+)},\qquad
\tilde q:=\frac{u_-'\mc L_0u_+ - u_+'\mc L_0 u_-}{W(u_-,u_+)}. \]
By construction, we have
$\tilde{\mc L} u_\pm=\mc L_0u_\pm+\tilde pu_\pm'+\tilde qu_\pm=0$.
In order to estimate $p-\tilde p$, we first note that
\[ u_+(x)\mc L_0 u_-(x)-u_-(x)\mc L_0 u_+(x)=O((1-x^2)^{-4}) \] since the most singular terms cancel. Consequently,
\[ P_1(x):=(1-x^2)^4[u_+(x)\mc L_0u_-(x) - u_-(x)\mc L_0 u_+(x) ] \]
is a polynomial of degree at most $66$.
Furthermore, recall that
\begin{align*}
 p(x)&=\partial_3 \Phi\big (x,\gT(x), \gT'(x)\big )+\frac{8x}{1-x^2}
 =\frac{\Phi_1(x,\gT(x))+2\Phi_2(x,\gT(x))\gT'(x)}
 {\Psi(x,y)}+\frac{8x}{1-x^2} \\
&=2\gT'(x)\frac{\hat\Phi_2(x,\gT(x))}{\hat\Psi(x,\gT(x))}
+\frac{1}{1-x^2}\frac{\Phi_1(x,\gT(x))+8x\hat\Psi(x,\gT(x))}{\hat\Psi(x,\gT(x))},
\end{align*}
where we use the notation
\[ \hat\Psi(x,y)=\frac{\Psi(x,y)}{1-x^2},\qquad \hat\Phi_k(x,y)=\frac{\Phi_k(x,y)}{1-x^2}. \]
From Eqs.~\eqref{def:Phik}, \eqref{def:Psi} it follows that $\hat\Psi$ and $\hat\Phi_2$ are polynomials.
Moreover, we have
\[ \Phi_1(x, y)+8x\hat\Psi(x,y)=0 \]
for $x\in \{-1,1\}$
and this shows that $p$ is of the form
$p(x)=\frac{P_2(x)}{P_3(x)}$
where
\[ P_2(x):=2\gT'(x)\hat\Phi_2(x,\gT(x))+\frac{\Phi_1(x,\gT(x))+8x\hat\Psi(x,\gT(x))}{1-x^2} \]
is a polynomial of degree at most $263$ and
$P_3(x):=\hat\Psi(x,\gT(x))$.
Recall that $P_3$ has no zeros on $[-1,1]$ and $\deg P_3\leq 264$.
Consequently, we obtain
\[ p-\tilde p=\frac{P_2}{P_3}-\frac{P_1}{W_0}=\frac{P_2 W_0-P_1 P_3}{P_3W_0}. \]
In order to estimate this expression, we proceed as in the proof of Proposition \ref{prop:Rem}.
First, we straighten the denominator, i.e., we try to find an approximation to $\frac{1}{W_0P_3}$ as a truncated Chebyshev expansion.
To improve the numerical convergence, it is advantageous to multiply the numerator and denominator by the polynomial $(\frac{13}{10}-x^2)^8$ (this factor is found empirically).
Consequently, we write
$p-\tilde p=\frac{P_4}{P_5}$
where 
\[ P_4(x)=(\tfrac{13}{10}-x^2)^8[P_2(x)W_0(x)-P_1(x)P_3(x)],\qquad
P_5(x)=(\tfrac{13}{10}-x^2)^8 P_3(x)W_0(x). \]
Note that $P_4$ and $P_5$ are polynomials with rational coefficients and
$\deg P_4\leq 346$, $\deg P_5\leq 341$.
Next, we obtain an approximation to $1/P_5$ of the form
\[ \frac{1}{P_5(x)}\approx \sum_{n=0}^{30}r_n T_n(x)=:R(x) \]
where the coefficients $(r_n)_{n=1}^{30}\subset \Q$, obtained by a pseudospectral method, are given in Table \ref{tab:r1} and $r_0=-\frac{623}{23}$.
We write $p-\tilde p=\frac{RP_4}{RP_5}$ and note that $\deg (RP_4)\leq 376$, $\deg (RP_5)\leq 371$.
We re-expand $RP_4$ and $RP_5$ as
\[ RP_4=\sum_{n=0}^{376}p_{4,n}T_n,\qquad RP_5=\sum_{n=0}^{376} p_{5,n}T_n \]
by solving the linear equations
\begin{align*} \sum_{n=0}^{376}p_{4,n}T_n(x_k)&=RP_4(x_k), \qquad
\sum_{n=0}^{376}p_{5,n}T_n(x_k)=RP_5(x_k)
\end{align*}
for
$x_k=-\tfrac12+\tfrac{k}{376}$ and
$k=0,1,\dots,376$. 
This yields the bound
\[ \|(RP_5)'\|_{L^\infty}\leq \sum_{n=0}^{376} |p_{5,n}|\|T_n'\|_{L^\infty}
\leq \sum_{n=0}^{376}n^2 |p_{5,n}|\leq 17 \]
and from Lemma \ref{lem:estf} with $N=1000$ we infer
\[ \min_{[-1,1]}RP_5\geq \min_{\Omega_N}RP_5-\tfrac{2}{N}\|(RP_5)'\|_{L^\infty}
\geq \tfrac{98}{100}-\tfrac{34}{1000}\geq \tfrac{94}{100}. \]
Consequently, we find
\[ \|p-\tilde p\|_{L^\infty}=\left \|\tfrac{RP_4}{RP_5}\right \|_{L^\infty}
\leq \tfrac{100}{94}\sum_{n=0}^{376}|p_{4,n}|\leq \tfrac{3}{100}. \]
The bound for $q-\tilde q$ is proved analogously.
\end{proof}

\begin{proposition}
\label{prop:fsbounds}
The approximate fundamental system $\{u_-,u_+\}$ satisfies
the bounds
\begin{align*} 
|u_-(x)|\int_x^1 \frac{|u_+(y)|}{|W(y)|}\d y+|u_+(x)|\int_{-1}^x \frac{|u_-(y)|}{|W(y)|}\d y&\leq \tfrac{7}{10} \\
|u_-'(x)|\int_x^1 \frac{|u_+(y)|}{|W(y)|}\d y+
|u_+'(x)|\int_{-1}^x \frac{|u_-(y)|}{|W(y)|}\d y&\leq \tfrac12 
\end{align*}
for all $x\in (-1,1)$, where $W(y):=W(u_-,u_+)(y)$.
\end{proposition}

\begin{proof}
As before, we write $u_\pm(x)=(1\pm x)^{-3}w_\pm(x)$ and recall that $w_\pm$ are polynomials of degree $30$, see Definition \ref{def:fs}.
First, we obtain an approximation to $1/W_0$, where $W(x)=(1-x^2)^{-4} W_0(x)$, see Proposition \ref{prop:fs}.
By employing the usual pseudospectral method, 
we find
\[ \frac{1}{W_0(x)}\approx \sum_{n=0}^{22}r_n T_n(x)=:R(x) \]
with the coefficients $(r_n)_{n=0}^{22}\subset \Q$ given in Table \ref{tab:r2}.
Next, we note that 
\[ |\psi_{-,n}'(x)|\leq |T_n'(x)|+|T_n'(-1)|+2|T_n(-1)|\leq 2n^2+2 \]
for all $x\in [-1,1]$,
see Eq.~\eqref{def:psi},
and thus,
\[ \|w_-'\|_{L^\infty}\leq \sum_{n=1}^{30}|c_{-,n}|\|\psi_{-,n}'\|_{L^\infty}
\leq 2\sum_{n=1}^{30}(n^2+1)|c_{-,n}|\leq 60. \]
Consequently, Lemma \ref{lem:estf} with $N=600$ yields
\[ \min_{[-1,1]}w_-\geq \min_{\Omega_N}w_- - \tfrac{2}{N}\|w_-'\|_{L^\infty}\geq \tfrac{7}{10}-\tfrac{1}{5}=\tfrac12 \]
and in particular, $w_->0$.
Analogously, we see that $w_+>0$ on $[-1,1]$.
Furthermore, from the proof of Proposition \ref{prop:fs} we recall that $W_0<0$ on $[-1,1]$.
Consequently, we find
\begin{align*}
A(x):&=|u_-(x)|\int_x^1 \frac{|u_+(y)|}{|W(y)|}\d y+|u_+(x)|\int_{-1}^x \frac{|u_-(y)|}{|W(y)|}\d y \\
&=-\frac{w_-(x)}{(1-x)^3}\int_x^1 (1-y)^4(1+y)\frac{R(y) w_+(y)}{R(y) W_0(y)}\d y \\
&\quad -\frac{w_+(x)}{(1+x)^3}\int_{-1}^x (1+y)^4(1-y)\frac{R(y) w_-(y)}{R(y) W_0(y)}\d y.
\end{align*}
Note that $RW_0$ is a polynomial of degree at most $22+61=83$, see the proof of Proposition \ref{prop:fs}.
We re-expand $RW_0$ by solving the linear system
\[ \sum_{n=0}^{83} a_n T_n(x_k)=R(x_k)W_0(x_k), \quad x_k=-\tfrac12+\tfrac{k}{83},\quad k=0,1,\dots,83 \]
over $\Q$, which yields the estimate
\[ \|(RW_0)'\|_{L^\infty}\leq \sum_{n=0}^{83}n^2 |a_n|\leq 3. \]
Thus, from Lemma \ref{lem:estf} with $N=600$ we infer
\[ \min_{[-1,1]}RW_0\geq \min_{\Omega_N}RW_0-\tfrac{2}{N}\|(RW_0)'\|_{L^\infty} 
\geq \tfrac{99}{100}-\tfrac{1}{100}=\tfrac{98}{100} \]
and this yields
\begin{align*}
A(x)&\leq \tfrac{100}{98}\left [\frac{w_-(x)}{(1-x)^3}I_+(x)+\frac{w_+(x)}{(1+x)^3}I_-(x)\right ],
\end{align*}
where
\begin{align}
\label{def:Ipm}
I_-(x)&:=\int_{-1}^x (1+y)^4(1-y)[-R(y)] w_-(y)\d y \nonumber \\
I_+(x)&:=\int_x^1 (1-y)^4(1+y)[-R(y)]w_+(y)\d y.
\end{align}
The integrands of $I_\pm$ are polynomials and hence, $I_\pm$ can be computed explicitly.
More precisely, we write
\[ P_\pm(y):=(1\mp y)^4(1\pm y)[-R(y)]w_\pm(y) \]
and note that $\deg P_\pm \leq 57$.
Consequently, we may re-expand $P_\pm$ as
$P_\pm(y)=\sum_{n=0}^{57}p_{\pm,n}y^n$
by solving the linear systems
\[ \sum_{n=0}^{57} p_{\pm,n}x_k^n=P_\pm(x_k),\quad x_k=-\tfrac12+\tfrac{k}{57},\quad k=0,1,2,\dots,57 \]
over $\Q$.
From this we obtain the explicit expressions
\begin{align*}
I_-(x)&=\sum_{n=0}^{57}\frac{p_{-,n}}{n+1}x^{n+1}-\sum_{n=0}^{57}\frac{p_{-,n}}{n+1}(-1)^{n+1} \\
I_+(x)&=\sum_{n=0}^{57}\frac{p_{+,n}}{n+1}-\sum_{n=0}^{57}\frac{p_{+,n}}{n+1}x^{n+1}.
\end{align*}
Furthermore, directly from Eq.~\eqref{def:Ipm} we see that $I_\pm(x)=O((1\mp x)^5)$.
Consequently, 
\[ P(x):=\frac{w_-(x)}{(1-x)^3}I_+(x)+\frac{w_+(x)}{(1+x)^3}I_-(x) \]
is a polynomial of degree at most $85$.
Thus, another re-expansion yields the Chebyshev representation
$P(x)=\sum_{n=0}^{85}p_n T_n(x)$
and we obtain the bound
\[ \|P'\|_{L^\infty}\leq \sum_{n=0}^{85} n^2 |p_n|\leq 3. \]
Consequently, Lemma \ref{lem:estf} with $N=1000$ yields
\[ A(x)\leq \tfrac{100}{98}\|P\|_{L^\infty}\leq \tfrac{100}{98}\left (\max_{\Omega_N}|P|+\tfrac{2}{N}\|P'\|_{L^\infty}\right ) 
\leq \tfrac{100}{98}\left (\tfrac{591}{1000}+\tfrac{6}{1000} \right )\leq \tfrac{7}{10}. \]

To prove the second bound, we set $Q_{\pm}(x):=u_\pm'(x)I_{\mp}(x)$ and note that
\[ u_\pm'(x)=\frac{w_\pm'(x)}{(1\pm x)^3}\mp 3\frac{w_\pm(x)}{(1\pm x)^4}. \]
Consequently, $Q_\pm$ are polynomials with $\deg Q_\pm\leq 84$ and a Chebyshev re-expansion yields
\[ \|Q_-'\|_{L^\infty}+\|Q_+'\|_{L^\infty}\leq 20. \]
Thus, from Lemma \ref{lem:estf} with $N=800$ we infer\footnote{Strictly speaking, a slight variant of Lemma \ref{lem:estf} is necessary here since the function $|Q_-|+|Q_+|$ is only piecewise $C^1$.}
\begin{align*} \max_{[-1,1]} \left (|Q_-|+|Q_+|\right )&\leq \max_{\Omega_N}\left (|Q_-|+|Q_+|\right )+\tfrac{2}{N}\left (\|Q_-'\|_{L^\infty}+\|Q_+'\|_{L^\infty}\right ) \\
&\leq \tfrac{41}{100}+\tfrac{5}{100}=\tfrac{46}{100}
\end{align*}
which implies
\begin{align*}
|u_-'(x)|\int_x^1 \frac{|u_+(y)|}{|W(y)|}\d y+
|u_+'(x)|\int_{-1}^x \frac{|u_-(y)|}{|W(y)|}\d y&\leq
\tfrac{100}{98}\left (|u_-'(x)I_+(x)|+|u_+'(x)I_-(x)|\right ) \\
&=\tfrac{100}{98}\left (|Q_-(x)|+|Q_+(x)|\right ) \\
&\leq \tfrac{100}{98}\tfrac{46}{100}\leq \tfrac12
\end{align*}
for all $x\in (-1,1)$.
\end{proof}

\subsection{Construction of the Green function}

Based on Proposition \ref{prop:fs} we can now invert the operator $\tilde{\mc L}$.
A solution of the equation $\tilde{\mc L}u=f \in L^\infty(-1,1)$ is given by
\[ u(x)=\int_{-1}^1 G(x,y)f(y)\d y, \]
with the Green function
\[ G(x,y)=\frac{1}{W(u_{-},u_+)(y)}\left \{\begin{array}{lr}u_{-}(x)u_+(y) & x\leq y \\
u_+(x)u_{-}(y) & x\geq y \end{array} \right . . \]
In fact, this is the \emph{unique} solution that belongs to $L^\infty(-1,1)$.
Consequently, we have
\[ \tilde{\mc L}^{-1}f(x)=\int_{-1}^1 G(x,y)f(y)\d y. \]
The bounds from Proposition \ref{prop:fsbounds} immediately imply the following estimate.

\begin{corollary}
\label{cor:CL}
We have the bound
\[ \|\tilde{\mc L}^{-1}f\|_{W^{1,\infty}}\leq \|f\|_{L^\infty} \]
for all $f\in L^\infty(-1,1)$.
\end{corollary}

\begin{proof}
By definition we have
\[ \tilde{\mc L}^{-1}f(x)=u_-(x)\int_x^1 \frac{u_+(y)}{W(y)}f(y)\d y
+u_+(x)\int_{-1}^x \frac{u_-(y)}{W(y)}f(y)\d y \]
and thus,
\[ (\tilde{\mc L}^{-1} f)'(x)=u_-'(x)\int_x^1 \frac{u_+(y)}{W(y)}f(y)\d y
+u_+'(x)\int_{-1}^x \frac{u_-(y)}{W(y)}f(y)\d y, \]
where $W(y)=W(u_-,u_+)(y)$.
Consequently, from Proposition \ref{prop:fsbounds} we infer 
\begin{align*} \|\tilde{\mc L}^{-1}f\|_{W^{1,\infty}}= \left (
\|(\tilde{\mc L}^{-1}f)'\|_{L^\infty}^2+\|\tilde{\mc L}^{-1}f\|_{L^\infty}^2 \right )^{1/2} 
\leq \sqrt{(\tfrac12)^2+(\tfrac{7}{10})^2}\,\|f\|_{L^\infty}\leq \|f\|_{L^\infty}.
\end{align*}
\end{proof}

\section{Linear stability of the Skyrmion}

\noindent Now we are ready to conclude the proof of Theorem \ref{thm:main}.

\subsection{The main contraction argument}
 
Recall that we aim for solving the equation $\mc R(\gT+\delta)=0$, i.e.,
\[ \mc L\delta=-\mc R(\gT)-\mc N(\delta), \]
see Proposition \ref{prop:CN}.
We rewrite this equation as
\[ \tilde{\mc L}\delta=-\mc R(\gT)+(\tilde{\mc L}-\mc L)\delta-\mc N(\delta) \]
and apply $\tilde{\mc L}^{-1}$, which yields
\[ \delta=-\tilde{\mc L}^{-1}\mc R(\gT)+\tilde{\mc L}^{-1}(\tilde{\mc L}-\mc L)\delta
-\tilde{\mc L}^{-1}\mc N(\delta)=:\mc K(\delta) \]
Thus, our goal is to prove that $\mc K$ has a fixed point.

\begin{lemma}
\label{lem:contr}
Let $X:=\{u\in C^1[-1,1]: \|u\|_{W^{1,\infty}}\leq \frac{1}{150}\}$.
Then $\mc K$ has a unique fixed point in $X$.
\end{lemma}

\begin{proof}
From Propositions \ref{prop:Rem}, \ref{prop:CN}, \ref{prop:fs}, and Corollary \ref{cor:CL} we obtain the estimate
\begin{align*}
\|\mc K(u)\|_{W^{1,\infty}}&\leq \|\mc R(\gT)\|_{L^\infty}
+\|\mc Lu-\tilde{\mc L}u\|_{L^\infty}
+\|\mc N(u)\|_{L^\infty} \\
&\leq \tfrac{1}{500}+\|p-\tilde p\|_{L^\infty}\|u'\|_{L^\infty}+\|q-\tilde q\|_{L^\infty}\|u\|_{L^\infty}+39\|u\|_{W^{1,\infty}}^2 \\
&\leq \tfrac{1}{500}+\tfrac{3}{100}\tfrac{1}{150}+\tfrac{1}{20}\tfrac{1}{150}+39(\tfrac{1}{150})^2 \\
&\leq \tfrac{1}{150}.
\end{align*}
Consequently, $\mc K(u)\in X$ for all $u\in X$.
Furthermore,
\begin{align*}
\|\mc K(u)-\mc K(v)\|_{W^{1,\infty}}&\leq \|(\mc L-\tilde{\mc L})(u-v)\|_{L^\infty}
+\|\mc N(u)-\mc N(v)\|_{L^\infty} \\
&\leq \|p-\tilde p\|_{L^\infty}\|u'-v'\|_{L^\infty}+\|q-\tilde q\|_{L^\infty}\|u-v\|_{L^\infty}+39\tfrac{2}{150}\|u-v\|_{W^{1,\infty}} \\
&\leq \left (\tfrac{3}{100}+\tfrac{1}{20}+\tfrac{78}{150}\right )\|u-v\|_{W^{1,\infty}} \\
&=\tfrac35 \|u-v\|_{W^{1,\infty}}
\end{align*}
for all $u,v\in X$.
Thus, the claim follows from the contraction mapping principle.
\end{proof}

Finally, we obtain the desired approximation to the Skyrmion.

\begin{corollary}
\label{cor:F0}
There exists a $\delta \in C^1[-1,1]$ with $\|\delta\|_{W^{1,\infty}}\leq \tfrac{1}{150}$ such that the Skyrmion is given by
\[ F_0(r)=2\arctan\left (r(1+r)\left (\gT\left (\frac{r-1}{r+1}\right ) +\delta\left (
\frac{r-1}{r+1}\right )\right )\right ). \]
\end{corollary}

\begin{proof}
By construction, Lemma \ref{lem:contr}, and standard ODE regularity theory, there exists a $\delta$ with the stated properties such that $F_0$ is a smooth solution to the original Skyrmion equation \eqref{eq:sky}.
Obviously, we have $F_0(0)=0$ and from $\gT(x)\in [\frac12,\frac32]$ for all $x\in [-1,1]$, see Proposition \ref{prop:Rem}, we infer $\lim_{r\to\infty}F_0(r)=\pi$.
Since the Skyrmion is the unique solution of Eq.~\eqref{eq:sky} with these boundary values \cite{McLTro91}, the claim follows.
\end{proof}

\subsection{Spectral stability}
Recall that the linear stability of the Skyrmion is governed by the Schr\"odinger operator
\[ \mc Af(r)=-f''(r)+\frac{2}{r^2}f(r)+V(r)f(r) \]
on $L^2(0,\infty)$, 
where the potential is given by
\[ V=-4a^2\frac{1+3a^2+3a^4}{(1+2a^2)^2}, \qquad a(r)=\frac{\sin F_0(r)}{r}. \]
From Corollary \ref{cor:F0} and the identity $\sin(2\arctan y)=\frac{2y}{1+y^2}$ we obtain
\[ a(r)=\frac{2(1+r)\left [\gT \left(\frac{r-1}{r+1}\right )+\delta\left (\frac{r-1}{r+1}\right )\right ]}
{1+r^2(1+r)^2 \left [\gT \left(\frac{r-1}{r+1}\right )+\delta\left (\frac{r-1}{r+1}\right )\right ]^2}. \]
Furthermore, from $\|\delta\|_{L^\infty}\leq \frac{1}{150}$ and $\gT(x)\in [\frac12,\frac32]$ for all $x\in [-1,1]$, see Proposition \ref{prop:Rem}, we infer the bounds
\begin{align*}
 |a(r)|&\leq \frac{2(1+r)\left [\gT \left(\frac{r-1}{r+1}\right )+\frac{1}{150}\right ]}
{1+r^2(1+r)^2 \left [\gT \left(\frac{r-1}{r+1}\right )-\frac{1}{150}\right ]^2}=:A(r) \\
|a(r)|&\geq \frac{2(1+r)\left [\gT \left(\frac{r-1}{r+1}\right )-\frac{1}{150}\right ]}
{1+r^2(1+r)^2 \left [\gT \left(\frac{r-1}{r+1}\right )+\frac{1}{150}\right ]^2}=:B(r)
\end{align*}
Consequently, we obtain the estimate
\[ |V|\leq 4A^2\frac{1+3A^2+3A^4}{(1+2B^2)^2}. \]

\begin{lemma}
\label{lem:intV}
We have the bound
\[ \int_0^\infty r^7 |V(r)|^4\d r\leq 130.\]
\end{lemma}

\begin{proof}
By employing the techniques introduced before, it is straightforward to obtain the stated estimate. More precisely, we introduce the new integration variable $x\in [-1,1]$, given by $r=\frac{1+x}{1-x}$,
and write
\[ \int_0^\infty r^7 |V(r)|^4\d r\leq \int_0^\infty r^7 \left [4A(r)^2\frac{1+3A(r)^2+3A(r)^4}{(1+2B(r)^2)^2}\right ]^4\d r=\int_{-1}^1 \frac{P(x)}{Q(x)}\d x, \]
where $P$ and $Q$ are polynomials with rational coefficients.
As before, by a pseudospectral method, we construct a truncated Chebyshev expansion $R(x)$ of $1/Q(x)$.
Next, by a Chebyshev re-expansion we obtain an estimate for $\|(RQ)'\|_{L^\infty}$
and Lemma \ref{lem:estf} yields a lower bound on $\min_{[-1,1]}RQ$ which is close to $1$.
From this we find
\[ \int_{-1}^1 \frac{P(x)}{Q(x)}\d x=\int_{-1}^1 \frac{R(x)P(x)}{R(x)Q(x)}\d x
\leq \frac{1}{\min_{[-1,1]}RQ}\int_{-1}^1 R(x)P(x)\d x \]
and the last integral can be evaluated explicitly since the integrand is a polynomial.
\end{proof}

We can now conclude the main result.

\begin{proof}[Proof of Theorem \ref{thm:main}]
From Lemma \ref{lem:intV} we obtain
\[ 3^{-7}\frac{3^3\Gamma(8)}{4^4\Gamma(4)^2}\int_0^\infty r^7 |V(r)|^4\d r\leq \frac{2275}{2592}<1.  \]
Consequently, the GGMT bound, see Appendix \ref{app:GGMT}, implies that $\mc A$ has no eigenvalues.
\end{proof}

\begin{appendix}

\section{The GGMT bound}
\label{app:GGMT}

\noindent Consider $H=-\Delta + V$ in $\R^3$ where $V\in L^1\cap L^\infty(\R^3)$ (say) and radial. 
The GGMT bound~\cite{GlaGroMarThi75} is as follows (see also~\cite{GlaGroMar78}). We restrict ourselves to a smaller range of $p$ than necessary since
it is technically easier and sufficient.  

\begin{theorem}\label{thm:GGMT} 
Write $V=V_+ - V_-$ where $V_{\pm}\ge0$. 
For any $\f32\le p<\I$, if 
\EQ{\label{eq:GGMT}
\f{(p-1)^{p-1} \Gamma(2p)}{p^p \Gamma^2(p)} \int_{0}^\infty r^{2p-1} V_{-}^p(r)\, dr <1
}
then $H$ has no negative eigenvalues. Furthermore, zero energy is neither an eigenvalue nor a 
resonance. 
\end{theorem}
\begin{proof}
Suppose $H$ has negative spectrum. Then there exists a ground state, $H\psi=E\psi$ with $\psi\in H^2(\R^3)$, $\|\psi\|_2=1$, 
and radial, $E<0$. 
So
\EQ{\label{eq:0}
\LR{ H\psi,\psi} <0
}
which implies in particular that for any $\al\in \R$, 
\EQ{ \label{eq:V-}
\int_{\R^3}  |\nabla\psi(x)|^2\, dx &< \int_{\R^3}  V_-(x) |\psi(x)|^2\, dx \\
&\leq \| r^\al V_-\|_p \| r^{-\f{\al}{2}} \psi\|_{2q}^2 
}
by H\"older, $\f1p+\f1q=1$ (which is only meaningful if the right-hand side is finite).  We set $p(2-\al)=3$, $q(1+\al)=3$, which requires that $-1\leq \al\leq 2$. 
In fact, $\I\geq p\ge \f32$ means precisely that $2\geq \al\geq0$, and $1\le q\le 3$. 
Set
\EQ{\label{eq:mu q}
\mu_q := \inf_{\psi\in H^1_{\mathrm rad}\setminus\{0\}} \f{ \| \nabla \psi\|_2^2}{\| r^{\f{q-3}{2q}} \psi \|_{2q}^2}
}
Note that the denominator here is always a positive finite number. Indeed, it suffices to check this for $q=1$ and $q=3$, respectively. 
This amounts to  
\[
\| r^{-1} \psi\|_2 + \|\psi\|_6 \leq C\|\nabla \psi\|_2 \qquad \forall\; \psi\in H^1(\R^3)
\]
which is true by the Hardy and Sobolev inequalities.  By Lemma~\ref{lem:calV}, $\mu_q>0$ and its value can be explicitly computed.  Thus,  by \eqref{eq:V-},
\[
\|\nabla\psi\|_2^2 \leq \mu_q^{-1} \| r^\al V_-\|_p \|\nabla\psi\|_2^2
\]
which is a contradiction of $\mu_q^{-1} \| r^\al V_-\|_p<1$, the latter being precisely condition~\eqref{eq:GGMT}. 

It remains to discuss the case where $H$ has no negative spectrum but a zero eigenvalue or a zero resonance. If $0$ is an eigenvalue, then we have a solution $\psi\in H^2$ of
\[
-\Delta \psi = V\psi
\]
which means that 
\[
\psi(x) = - \f{1}{4\pi} \int_{\R^3} \f{V(y)\psi(y)}{|x-y|}\, dy
\]
If $\int V\psi\ne0$, then $\psi(x)\simeq |x|^{-1}$ for large $x$, which is not $L^2$.  So $\int V\psi=0$ and $\psi(x)=O(|x|^{-2})$ as $x\to\I$. 
One has $\LR{H\psi,\psi}=0$ instead of~\eqref{eq:0}. 
Replacing $H$ with $H_\eps=H-\eps e^{-|x|^2}$ for small $\eps>0$ we conclude that 
\[
\LR{ H_\eps \psi,\psi}<0
\] and
$H_\eps$ therefore has negative spectrum,
while~\eqref{eq:GGMT} still holds for small~$\eps$. By the previous case, this gives a contradiction. 

If $0$ is a resonance, this means that there is a solution $\psi \in H^2_{\mathrm{loc}}(\R^3)$ with $\psi(x)\simeq |x|^{-1}$ as $x\to\I$
(and by the reasoning above this holds if and only if  $\int V\psi\ne 0$).   In particular, since $\nabla \psi \in L^2$ and since $\int V\psi^2$ is absolutely convergent, we still arrive
at the conclusion that  $\LR{H\psi,\psi}=0$. Substituting $H_\eps$ for $H$ as above again gives a contradiction. 
To be precise, we evaluate the quadratic form of $H_\eps$ on the functions $$\psi_R(x):=\chi(x/R) \psi(x)$$
where $\chi$ is a standard bump function of compact support and equal to~$1$ on the unit ball. Sending $R\to\I$ then shows that 
$H_\eps$ has negative spectrum. 
\end{proof}

The following lemma establishes the constant $\mu_q$ in the previous proof.
The variational problem \eqref{eq:mu q} is invariant under a two-dimensional group of  symmetries: $S_{\xi,\eta}(\psi)(x) =  e^{\xi}\psi(e^{\eta} x)$, $\xi,\eta\in\R$. 
The scaling of the independent variable leads to a loss of compactness. To make it easier to apply the standard 
methods of concentration-compactness, we employ the same change of variables as in~\cite{GlaGroMarThi75}. 

\begin{lemma}\label{lem:calV}
For $1< q\le 3$ we have 
\[
\mu_q = \f{p}{p-1} \Big[ 4\pi \f{(p-1)\Gamma^2(p)}{\Gamma(2p)}\Big]^{\f1p}, 
\]
with $p$ the dual exponent to $q$. Equality in \eqref{eq:mu q} is attained by the radial functions
\[
\psi_q(x) = \f{a}{\big(1+b r^{\f{1}{p-1}}\big)^{p-1}}
\]
where $a,b>0$ are arbitrary. 
\end{lemma}
\begin{proof}
We begin with the following claim
\EQ{\label{eq:claim}
\mu_q = \inf_{\fy\in H^1(\R),\fy\ne0} (4\pi)^{\f1p}\f{ \int_{-\I}^\I \big( \fy'^2 +\f14 \fy^2)(x)\, dx}{\Big( \int_{-\I}^\I \fy^{2q}(x)\, dx\Big)^{\f1q}} 
}
To prove it, first note that we may  take the infimum in \eqref{eq:mu q} over  radial functions $\psi \in C^1(\R^3)$ of compact support. 
For this we use that $1\le  q\le 3$ to control the denominator by the $\dot H^1(\R^3)$ norm. 
Then set $\fy(x)= \sqrt{r}\psi(r)$, $r=e^x$ and calculate
\EQ{ \nonumber 
\| \nabla \psi\|_2^2 &= 4\pi \int_{-\I}^\I \big( \fy'^2 +\f14 \fy^2)(x)\, dx \\
 \| r^{\f{q-3}{2q}} \psi \|_{2q}^2 &= (4\pi)^{\f1q} \Big(  \int_{-\I}^\I  \fy^{2q}(x)\, dx\Big)^{\f1q}
}
Note that $\fy(x)\to0$ as $x\to\pm\I$ (exponentially as $x\to-\I$, and identically vanishing for large $x>0$). This gives~\eqref{eq:claim} 
by density.  

Let $\fy_n\in H^1(\R)$ be a minimizing sequence for~\eqref{eq:claim} with $\|\fy_n\|_{2q}=1$. Clearly, $\fy_n$ is bounded in $H^1(\R)$ and by Sobolev embedding, it follows that $\mu_q>0$. 
By the concentration compactness method, see Proposition~3.1
in~\cite{HmiKer05}, there exist $V_j\in H^1(\R)$ for all $j\ge 1$, and $x_{j,n}\in \R$ such that  (everything up to passing to subsequences)
\EQ{\nn 
|x_{j,n} - x_{k,n} | &\to \I   \text{\ \ \ \ for all $j\ne k$ as\ \ } n\to\I \\
\fy_n &= \sum_{j=1}^\ell V_j(\cdot-x_{j,n}) + g_{n,\ell} \text{\ \ \ \ for all\ \ } \ell\ge1
}
where $$\limsup_{n\to\I} \|g_{n,\ell}\|_p \to 0$$ as $\ell\to\I$ for any $2<p<\I$. Moreover, 
\EQ{
\| \fy_n'\|_2^2 &=  \sum_{j=1}^\ell  \| V_j'\|_2^2 + \|g_{n,\ell}' \|_2^2 + o(1) \\
\| \fy_n \|_2^2 &=  \sum_{j=1}^\ell  \| V_j \|_2^2 + \|g_{n,\ell} \|_2^2 +o(1)
}
as $n\to\I$, and 
\[
1= \|\fy_n \|_{2q}^{2q} = \sum_{j=1}^\ell  \|V_j\|_{2q}^{2q} + o(1)
\]
as $n,\ell\to\I$.  To be precise, for any $\eps>0$ we may find $\ell$ such that 
\[
  \Big | 1- \sum_{j=1}^\ell  \|V_j\|_{2q}^{2q} \Big| <  \eps
\]
We have
\EQ{ \label{eq:part}
\| \fy_n'\|_2^2 + \f14 \|\fy_n\|_2^2 &\ge (4\pi)^{-\f1p}\mu_q \big( \sum_{j=1}^\ell \| V_j\|_{2q}^{2} + \|g_{n,\ell}\|_{2q}^2\big) -o(1)  \\
&\ge  (4\pi)^{-\f1p}\mu_q \big( \sum_{j=1}^\ell \| V_j\|_{2q}^{2q} + \|g_{n,\ell}\|_{2q}^{2q}\big)^{\f1q} -o(1) 
}
 If there were two nonzero profiles $V_j$, or if 
 $$
\limsup_{n\to\I} \|g_{n,\ell}\|_{2q} \not\to0
 $$
 as $\ell\to\I$, 
 then there exists $\delta>0$ (since $q>1$) so that 
 \EQ{ \nn
\| \fy_n'\|_2^2 + \f14 \|\fy_n\|_2^2 &\ge  (4\pi)^{-\f1p}\mu_q (1+\delta ) -o(1) 
}
as $n\to\I$, contradicting that $\fy_n$ is a minimizing sequence. So up to a translation, we may assume that $\fy_n$ is compact in $L^{2q}(\R)$ and
in fact that $\fy_n\to \fy_\I$ in $L^{2q}(\R)$. In particular, $\|\fy_\I \|_{2q}=1$.  Furthermore, we have the weak convergence $\fy_n'\rightharpoonup \fy_\I'$, $\fy_n\rightharpoonup \fy_\I$ 
in $L^2(\R)$ which implies that 
\[
\| \fy_\I'\|_2^2 + \f14 \|\fy_\I\|_2^2 \le \liminf_{n\to\I} \big(\| \fy_n'\|_2^2 + \f14 \|\fy_n\|_2^2\big) = (4\pi)^{-\f1p}\mu_q
\]
In conclusion, $\fy_n\to\fy_\I$ strongly in $H^1$, and $\fy_\I\in H^1(\R)\setminus\{0\}$  is a minimizer for $\mu_q$.    Passing absolute values onto $\fy_n$  we may assume that $\fy_\I\ge0$. 

The associated Euler-Lagrange equation is
\[
-2\fy_\I''+\f12 \fy_\I = k \fy_\I^{2q-1}
\]
first in the weak sense, but then in the classical one by basic regularity. Furthermore, $\fy>0$, $\fy_\I\in C^\I(\R)$, and $k>0$. Since $q>1$ we may
absorb the constant which leads to an exponentially decaying, positive smooth solution to the equation 
\[
-f''(x) +\f14 f(x) = f^{2q-1}(x)
\] 
By the phase portrait, such an $f$ is unique up to translation in $x$. It is given by the homoclinic orbit emanating from the origin and encircling the positive equilibrium. 
This homoclinic  orbit (and its reflection together with the origin) make up the algebraic curve
\EQ{\label{eq:hom}
-  f'^2 + \f14 f^2 = \f{1}{q} f^{2q}
}
The explicit form of the solution is obtained by integrating up the first order ODE \eqref{eq:hom} which leads to
\EQ{\label{eq:expl}
f(x-x_0) = \Big( \f{q}{4} \Big)^{\f{1}{2(q-1)}} \big( \cosh((q-1)x/2)\big)^{-\f{1}{q-1}}
}
where $x_0\in\R$.  Finally, 
\[
\mu_q^p = 4\pi \int_{-\I}^\I f(x)^{2q}\, dx
\]
with $f$ as on the right-hand side of \eqref{eq:expl}. Thus,
\EQ{\label{eq:muq cosh}
\mu_q^p = 4\pi \Big( \f{q}{4} \Big)^{\f{q}{q-1}} \int_{-\I}^\I  \big( \cosh((q-1)x/2)\big)^{-\f{2q}{q-1}}   \, dx
}
To proceed, we recall that for any $b>0$
\EQ{\nn
\int_{-\I}^\I (\cosh x)^{-b}\, dx &= 2^b \int_0^\I \f{u^{b-1}}{(1+u^2)^b}\, du = \f{\sqrt{\pi}\,\Gamma(b/2)}{\Gamma((b+1)/2)}
}
Inserting this into \eqref{eq:muq cosh} yields
\[
\mu_q^p = 4\pi \Big( \f{q}{4} \Big)^{p}\f{2}{q-1} \f{\Gamma(p)}{\Gamma(p+\f12)}
\]
Using $\Gamma(p)\Gamma(p+\f12)=2^{1-2p}\sqrt{\pi}\,\Gamma(2p)$ this turns into
\[
\mu_q^p = 4\pi \f{q^p}{q-1} \f{\Gamma(p)^2}{\Gamma(2p )}  = 4\pi \Big(  \f{p}{p-1} \Big)^p (p-1) \f{\Gamma(p)^2}{\Gamma(2p )}  
\]
which is what the lemma set out to prove. The minimizers are obtained by transforming \eqref{eq:expl} back to the original coordinates. 
\end{proof}

Theorem~\ref{thm:GGMT} is insufficient for linearized Skyrme. The reason being that the Helmholtz equation associated with the latter is of the form
\[
-\psi'' + \big(\frac{2}{r^2} + V(r))\psi = k^2\psi
\]
which has extra repulsivity coming from the $\frac{2}{r^2}$ potential. On the level of the Schr\"odinger equation in $\R^3$ this precisely amounts to 
restricting to angular momentum $\ell=1$. So we expect that a weaker condition on $V$ than the one stated in Theorem~\ref{thm:GGMT} will suffice. This is essential for our applications
to linearized Skyrme stability. 

In fact, as already noted in \cite{GlaGroMarThi75}, for general angular momentum $\ell>0$ we are faced with the minimzation problem 
which is obtained from~\eqref{eq:claim} by replacing $\frac14\fy^2$ with   $\frac14(2\ell+1)^2 \fy^2$.  However, the scaling $$\fy(x)=\fy_1((2\ell+1)x)$$ takes us back 
to the minimization problem~\eqref{eq:claim} with an extra factor of $(2\ell+1)^{1+\frac1q}$.  Recall that Theorem~\ref{thm:GGMT} is nothing other than $\mu_q^{-p} \| r^\alpha V_-\|_p^p<1$.   	Therefore, to exclude eigenfunctions and threshold resonances of angular momentum $\ell$ 
 condition \eqref{eq:GGMT} 
needs to be multiplied on the left by a factor of 
\[
(2\ell+1)^{-p(1+\frac1q)} = (2\ell+1)^{-(2p-1)}
\]
In the summary, the sufficient GGMT criterion for absence of bound states and threshold resonances in angular momentum $\ell$ reads
\EQ{\label{eq:GGMT*}
\f{(p-1)^{p-1} \Gamma(2p)}{(2\ell+1)^{2p-1}p^p \Gamma^2(p)} \int_{0}^\infty r^{2p-1} V_{-}^p(r)\, dr <1
}
for any $\frac32\le p<\infty$. 
For linear Skyrme stability we use this criterion with $\ell=1$ and $p=4$.

\section{Tables of expansion coefficients}

\begin{table}[ht]
\centering
\caption{Expansion coefficients for approximate Skyrmion}
\begin{tabular}{|c||c|c|c|c|c|c|c|c|c|c|c|c|c|c|c|c|c|c|c|c|}
\hline
$n$ & $2$ & $3$ & $4$ & $5$ & $6$ & $7$ & $8$\\ 
 \hline 
 $c_n$ & $\frac{13039}{72146}$ & $\frac{2909}{229801}$ & $-\frac{11670}{500821}$ & $-\frac{301}{39257}$ & $\frac{621}{122813}$ & $\frac{871}{221909}$ & $-\frac{42}{55481}$\\[2pt] \hline 
$n$ & $9$ & $10$ & $11$ & $12$ & $13$ & $14$ & $15$\\ 
 \hline 
 $c_n$ & $-\frac{64}{36275}$ & $-\frac{18}{77071}$ & $\frac{94}{139483}$ & $\frac{13}{40736}$ & $-\frac{31}{158602}$ & $-\frac{9}{42953}$ & $\frac{2}{100443}$\\[2pt] \hline  
$n$ & $16$ & $17$ & $18$ & $19$ & $20$ & $21$ & $22$\\ 
 \hline 
 $c_n$ & $\frac{11}{105144}$ & $\frac{2}{76485}$ & $-\frac{5}{121747}$ & $-\frac{5}{186976}$ & $\frac{1}{92977}$ & $\frac{2}{118683}$ & $\frac{1}{1805239}$\\[2pt]  \hline 
$n$ & $23$ & $24$ & $25$ & $26$ & $27$ & $28$ & $29$\\ 
 \hline 
 $c_n$ & $-\frac{1}{122146}$ & $-\frac{1}{317774}$ & $\frac{1}{332077}$ & $\frac{1}{377050}$ & $-\frac{1}{1689008}$ & $-\frac{1}{640158}$ & $-\frac{1}{3975308}$\\[2pt]  \hline 
$n$ & $30$ & $31$ & $32$ & $33$ & $34$ & $35$ & $36$\\ 
 \hline 
 $c_n$ & $\frac{1}{1402566}$ & $\frac{1}{2606123}$ & $-\frac{1}{4324868}$ & $-\frac{1}{3550160}$ & $\frac{1}{54392687}$ & $\frac{1}{6563655}$ & $\frac{1}{21696717}$\\[2pt]  \hline 
$n$ & $37$ & $38$ & $39$ & $40$ & $41$ & $42$ & $43$\\ 
 \hline 
 $c_n$ & $-\frac{1}{16289508}$ & $-\frac{1}{21329884}$ & $\frac{1}{86396283}$ & $\frac{1}{36311458}$ & $\frac{1}{128282128}$ & $-\frac{1}{128832209}$ & $-\frac{1}{196527234}$\\[2pt]
\hline 
\end{tabular}
\label{tab:c}
\end{table}

\begin{table}[ht]
\centering
\caption{Expansion coefficients for $u_-$}
\begin{tabular}{|c||c|c|c|c|c|c|c|c|c|c|c|c|c|c|c|c|c|c|c|c|}
\hline
$n$ & $1$ & $2$ & $3$ & $4$ & $5$ & $6$\\ 
 \hline 
 $c_{-,n}$ & $c_{-,1}$ & $\frac{5384}{2621}$ & $-\frac{711}{1909}$ & $\frac{417}{3424}$ & $\frac{18}{1817}$ & $\frac{2}{3169}$\\[2pt] \hline 
$n$ & $7$ & $8$ & $9$ & $10$ & $11$ & $12$\\ 
 \hline 
 $c_{-,n}$ & $-\frac{23}{3399}$ & $-\frac{22}{4655}$ & $\frac{4}{4097}$ & $\frac{7}{2589}$ & $\frac{2}{3607}$ & $-\frac{8}{6937}$\\[2pt] \hline 
$n$ & $13$ & $14$ & $15$ & $16$ & $17$ & $18$\\ 
 \hline 
 $c_{-,n}$ & $-\frac{3}{4310}$ & $\frac{1}{2886}$ & $\frac{2}{4135}$ & $-\frac{1}{90728}$ & $-\frac{1}{3865}$ & $-\frac{1}{11699}$\\[2pt] \hline 
$n$ & $19$ & $20$ & $21$ & $22$ & $23$ & $24$\\ 
 \hline 
 $c_{-,n}$ & $\frac{1}{9323}$ & $\frac{1}{11955}$ & $-\frac{1}{36563}$ & $-\frac{1}{18412}$ & $-\frac{1}{192414}$ & $\frac{1}{37653}$\\[2pt] \hline 
$n$ & $25$ & $26$ & $27$ & $28$ & $29$ & $30$\\ 
 \hline 
 $c_{-,n}$ & $\frac{1}{79523}$ & $-\frac{1}{119499}$ & $-\frac{1}{105631}$ & $-\frac{1}{1857125}$ & $\frac{1}{285782}$ & $\frac{1}{619658}$\\[2pt]
 \hline
 \end{tabular}
 \label{tab:cm}
 \end{table}
 
 \begin{table}[ht]
\centering
\caption{Expansion coefficients for $u_+$}
\begin{tabular}{|c||c|c|c|c|c|c|c|c|c|c|c|c|c|c|c|c|c|c|c|c|}
\hline
$n$ & $1$ & $2$ & $3$ & $4$ & $5$ & $6$\\ 
 \hline 
 $c_{+,n}$ & $c_{+,1}$ & $\frac{1371}{769}$ & $\frac{1734}{3319}$ & $\frac{230}{3431}$ & $-\frac{167}{6071}$ & $\frac{33}{5231}$\\[2pt] \hline 
$n$ & $7$ & $8$ & $9$ & $10$ & $11$ & $12$\\ 
 \hline 
 $c_{+,n}$ & $\frac{59}{4580}$ & $-\frac{19}{7202}$ & $-\frac{19}{2849}$ & $\frac{1}{13495}$ & $\frac{11}{3203}$ & $\frac{4}{4737}$\\[2pt] \hline 
$n$ & $13$ & $14$ & $15$ & $16$ & $17$ & $18$\\ 
 \hline 
 $c_{+,n}$ & $-\frac{7}{4481}$ & $-\frac{2}{2217}$ & $\frac{1}{1808}$ & $\frac{1}{1529}$ & $-\frac{1}{11699}$ & $-\frac{1}{2637}$\\[2pt] \hline 
$n$ & $19$ & $20$ & $21$ & $22$ & $23$ & $24$\\ 
 \hline 
 $c_{+,n}$ & $-\frac{1}{12409}$ & $\frac{1}{5625}$ & $\frac{1}{9479}$ & $-\frac{1}{16801}$ & $-\frac{1}{12593}$ & $\frac{1}{300485}$\\[2pt] \hline 
$n$ & $25$ & $26$ & $27$ & $28$ & $29$ & $30$\\ 
 \hline 
 $c_{+,n}$ & $\frac{1}{21636}$ & $\frac{1}{56764}$ & $-\frac{1}{51904}$ & $-\frac{1}{51451}$ & $\frac{1}{307476}$ & $\frac{1}{121058}$\\[2pt]
 \hline
 \end{tabular}
 \label{tab:cp}
 \end{table}

 \begin{table}[ht]
\centering
\caption{Expansion coefficients for approximation to $1/P_5$ (Proposition \ref{prop:fs})}
\begin{tabular}{|c||c|c|c|c|c|c|c|c|c|c|c|c|c|c|c|c|c|c|c|c|}
\hline
$n$ & $1$ & $2$ & $3$ & $4$ & $5$ & $6$\\ 
 \hline 
 $r_n$ & $-\frac{437}{24}$ & $-\frac{811}{20}$ & $-\frac{229}{17}$ & $-\frac{2391}{61}$ & $-\frac{397}{30}$ & $-\frac{178}{7}$\\[2pt] \hline 
$n$ & $7$ & $8$ & $9$ & $10$ & $11$ & $12$\\ 
 \hline 
 $r_n$ & $-\frac{184}{27}$ & $-\frac{518}{27}$ & $-\frac{98}{15}$ & $-\frac{1345}{114}$ & $-\frac{86}{31}$ & $-\frac{284}{39}$\\[2pt] \hline 
$n$ & $13$ & $14$ & $15$ & $16$ & $17$ & $18$\\ 
 \hline 
 $r_n$ & $-\frac{59}{24}$ & $-\frac{156}{35}$ & $-\frac{107}{106}$ & $-\frac{86}{37}$ & $-\frac{23}{31}$ & $-\frac{23}{16}$\\[2pt] \hline 
$n$ & $19$ & $20$ & $21$ & $22$ & $23$ & $24$\\ 
 \hline 
 $r_n$ & $-\frac{9}{26}$ & $-\frac{73}{110}$ & $-\frac{5}{27}$ & $-\frac{13}{32}$ & $-\frac{1}{9}$ & $-\frac{2}{11}$\\[2pt] \hline 
$n$ & $25$ & $26$ & $27$ & $28$ & $29$ & $30$\\ 
 \hline 
 $r_n$ & $-\frac{1}{24}$ & $-\frac{3}{29}$ & $-\frac{1}{29}$ & $-\frac{2}{33}$ & $-\frac{1}{62}$ & $-\frac{1}{47}$\\[2pt]
 \hline
 \end{tabular}
 \label{tab:r1}
 \end{table}

 \begin{table}[ht]
\centering
\caption{Expansion coefficients for approximation to $1/W_0$ (Proposition \ref{prop:fsbounds})}
\begin{tabular}{|c||c|c|c|c|c|c|c|c|c|c|c|c|c|c|c|c|c|c|c|c|}
\hline
 $n$ & $0$ & $1$ & $2$ & $3$ & $4$ & $5$ & $6$ & $7$\\ 
 \hline 
 $r_n$ & $-\frac{19}{69}$ & $\frac{11}{106}$ & $\frac{37}{103}$ & $-\frac{14}{107}$ & $-\frac{23}{128}$ & $\frac{7}{81}$ & $\frac{9}{109}$ & $-\frac{3}{67}$\\[2pt] \hline 
$n$ & $8$ & $9$ & $10$ & $11$ & $12$ & $13$ & $14$ & $15$\\ 
 \hline 
 $r_n$ & $-\frac{3}{79}$ & $\frac{2}{99}$ & $\frac{2}{111}$ & $-\frac{1}{124}$ & $-\frac{1}{114}$ & $\frac{1}{376}$ & $\frac{1}{233}$ & $-\frac{1}{1792}$\\[2pt] \hline 
$n$ & $16$ & $17$ & $18$ & $19$ & $20$ & $21$ & $22$ & $23$\\ 
 \hline 
 $r_n$ & $-\frac{1}{481}$ & $-\frac{1}{8890}$ & $\frac{1}{1025}$ & $\frac{1}{4569}$ & $-\frac{1}{2336}$ & $-\frac{1}{6718}$ & $\frac{1}{5790}$ & $0$\\[2pt]
 \hline
 \end{tabular}
 \label{tab:r2}
 \end{table}

\end{appendix}

\clearpage
\bibliography{skyrme}
\bibliographystyle{plain}

\end{document}